\RequirePackage[l2tabu, orthodox]{nag}		
\documentclass[10pt]{article}

\usepackage[T1]{fontenc}
\usepackage{lmodern}
\usepackage{natbib}
\usepackage{microtype}
\usepackage{amsmath}                        
\numberwithin{equation}{section}
\usepackage{amsfonts}
\usepackage{graphicx}                      
\usepackage[plainpages=false, pdfpagelabels]{hyperref} 
	\hypersetup{
		colorlinks   = true,
		citecolor    = RoyalBlue, 
		linkcolor    = RubineRed, 
		urlcolor     = Turquoise
	}
\usepackage[dvipsnames]{xcolor}
\usepackage[all]{xy}
\usepackage{amssymb}                        
\usepackage[mathscr]{eucal}                 
\usepackage{dsfont}													
\usepackage[paperwidth=8.5in,paperheight=11.00in,top=1.25in, bottom=1.25in, left=1.00in, right=1.00in]{geometry}
\usepackage{mathtools}                      
\mathtoolsset{showonlyrefs=true}            
\MakeRobust{\eqref}
\usepackage{microtype}											
\usepackage{amsthm}                         
\allowdisplaybreaks                         
\theoremstyle{plain}
\newtheorem{theorem}{Theorem}
\numberwithin{theorem}{section}
\newtheorem{proposition}[theorem]{Proposition}
\newtheorem{corollary}[theorem]{Corollary}
\theoremstyle{definition}

\newtheorem{example}[theorem]{Example}
\newtheorem{remark}[theorem]{Remark}

\newcommand{\<}{\langle}
\renewcommand{\>}{\rangle}
\renewcommand{\(}{\left(}
\renewcommand{\)}{\right)}

\newcommand\Cb{\mathds{C}}
\newcommand\Eb{\mathds{E}}
\newcommand\Fb{\mathds{F}}

\newcommand\Pb{\mathds{P}}

\newcommand\Rb{\mathds{R}}

\newcommand\Nb{\mathds{N}}

\newcommand\Fc{\mathscr{F}}

\newcommand\Nc{\mathscr{N}}
\newcommand\Oc{\mathscr{O}}

\newcommand\eps{\varepsilon}
\newcommand\om{\omega}
\newcommand\Om{\Omega}
\newcommand\sig{\sigma}

\newcommand\Gam{\Gamma}

\newcommand\Del{\Delta}

\newcommand\Fv{\mathbf{F}}

\newcommand\fh{\widehat{f}}

\newcommand\Hh{\widehat{H}}

\renewcommand\d{\partial}

\newcommand\ii{\mathtt{i}}
\newcommand\dd{\mathrm{d}}
\newcommand\ee{\mathrm{e}}

\newcommand\ko{\mathrm{ko}}
\newcommand\ki{\mathrm{ki}}
\newcommand\rb{\mathrm{rb}}

\newcommand{\Ib}[1]{\mathds{1}_{\{#1\}}}

\renewcommand\Re{\operatorname{Re}}
\renewcommand\Im{\operatorname{Im}}

\begin{document}

\title{Robust replication of barrier-style claims on price and volatility}

\author{
Peter Carr
\thanks{
Courant Institute, New York University.
\textbf{e-mail}: \url{pcarr@nyc.rr.com}
}
\and
Roger Lee
\thanks{
Department of Mathematics, University of Chicago.
\textbf{e-mail}: \url{rogerlee@math.uchicago.edu }
}
\and
Matthew Lorig
\thanks{
Department of Applied Mathematics, University of Washington.
\textbf{e-mail}: \url{mlorig@uw.edu}}
}

\date{This version: \today}

\maketitle

\begin{abstract}
We show how to price and replicate a variety of barrier-style claims written on the $\log$ price $X$ and quadratic variation $\<X\>$ of a risky asset.  Our framework assumes no arbitrage, frictionless markets and zero interest rates.  We model the risky asset as a strictly positive continuous semimartingale with an independent volatility process.  The volatility process may exhibit jumps and may be non-Markovian.  As hedging instruments, we use only the underlying risky asset, zero-coupon bonds, and European calls and puts with the same maturity as the barrier-style claim.  We consider knock-in, knock-out and rebate claims in single and double barrier varieties.
\end{abstract}


%
%

\section{Introduction}
Barrier options are the most liquid of the second generation options, (i.e., options whose payoffs are path-dependent).  In his landmark work, \cite{merton1973theory} first valued a down-and-out call in closed form when the underlying stock follows geometric Brownian motion. So long as the instantaneous volatility is a known function of the stock price and time, one can replicate any barrier claim by dynamically trading the stock and a zero-coupon bond.  If the volatility process is a continuous stochastic process driven by a second independent source of uncertainty, then one must also dynamically trade an option.  As with any hedge, the hedging strategy is invariant to the expected rate of return of the underlying stock.
\par
\cite{bowie-carr-1994} show how a static hedge in European options can be used to hedge down-and-out calls on futures in the Black model.  Essentially, the payoff of a down-and-out call with barrier $H$ can be replicated by buying a European call on the same underlying futures price with the same maturity $T$ and strike $K$ and also selling $K/H$ puts with strike $H^2/K$.  \cite{carr1998static} make clear that this static hedge works in any model with deterministic local volatility, provided that the volatility function is symmetric in the log
of the futures price relative to the barrier. We thus see a continuation of the pattern initiated by Merton in which hedging strategies are invariant to aspects of the statistical process.
\par
In \cite{carr1998static}, increments in the instantaneous volatility process are conditionally perfectly correlated with increments in the underlying futures price.  \cite{andreasen-2001} points out that the above hedge also works when increments in the instantaneous volatility process are conditionally independent of increments in the forward price. Similarly, \cite{bates1997skewness} observes that the above hedge works for ``Hull and White-type stochastic volatility processes.''  While Bates does not define this terminology, it seems reasonable to assert that both authors are assuming that the instantaneous volatility process is a diffusion, i.e., that the volatility process is continuous over time and has the strong Markov property, as in \cite{hullwhite1987}. Furthermore, as in \cite{hullwhite1987}, the volatility process should be autonomous in that its evolution coefficients refer only to volatility and time, but not the price of the underlying asset.
\par
\cite{pcs} make clear that these conditions are merely sufficient, but not necessary. 
The hedge described above for a down-and-out call works perfectly provided that there are no jumps over the barrier and the call and put have the same implied volatility at the first passage time to the barrier, if any. We refer to the latter condition as Put Call Symmetry (PCS), which was introduced to finance by \cite{bates1988} as a way to measure skewness. Hence, the bivariate process for the futures price and its volatility need not be Markov in itself. Furthermore, jumps in price and volatility can occur and increments in volatility can be correlated with returns, although some restrictions are necessary. As a result, we refer to these hedge strategies as \emph{semi-robust}.
\par
Barrier options are not the only path-dependent claims for which semi-robust hedges exist.  All of the barrier option hedges also extend to lookbacks. For lookbacks, the hedge is semi-static in that standard options are traded each time a new maximum is reached. Furthermore, assuming only no arbitrage, frictionless markets, zero interest rates, a positive continuous futures price process, and an independent
volatility process, \cite{rrvd} show how to replicate a variety of claims on 
the quadratic variation of returns experienced between initiation and a fixed maturity date. Their hedging instruments consist of the underlying futures and European options written on these futures at all strikes and with the same maturity as the claim.  In contrast to the hedges of barrier and lookback claims, their trading strategy in options is fully dynamic. Examples of claims whose payoffs can be spanned include volatility swaps and options on realized variance.
\par
The purpose of this manuscript is to synthesize the literature on semi-robust hedging of barrier claims and claims linked to quadratic variation. In particular, we show how to price and hedge claims on the $\log$ price $X$ and the quadratic variation $\<X\>$ of a risky asset subject to certain barrier events either occurring or not occurring.  
Examples of such claims include 
(i) a barrier start variance or volatility swap for which the non-negative payoff is the variance or volatility of $\log$ price experienced between the first passage time and a fixed maturity date,
(ii) a barrier start claim whose final payoff is the realized Sharpe ratio calculated between the first passage time and the fixed maturity date,
(iii) single and double barrier knock-out claims that, in the event no knock-out occurs, pay the product of powers and exponentials of log price and quadratic variation, and
(iv) a single barrier rebate claim that pays the product of powers and exponentials of quadratic variation if and when a barrier is reached prior to maturity.
\par
Our analysis makes the same assumptions as \cite{rrvd}. In particular, we consider a continuous time stochastic process for instantaneous volatility whose increments are uncorrelated with returns. Jumps in the volatility process are allowed and the evolution coefficients of the volatility process 
can refer to past or present values of the instantaneous volatility, time, and other variables as well, provided that they are independent of the futures price (i.e., non-Markovian dynamics are allowed for the volatility process).
Both foreign exchange and bond markets exhibit symmetric smiles, which, in a stochastic volatility setting, implies a volatility process that is uncorrelated with returns of the underlying \cite[Theorem 3.4]{pcs}.  Thus, our results are particularly relevant for these markets.
\par
The rest of this paper proceeds as follows.
In Section \ref{sec:model}, we introduce a general market model for a single risky asset $S = \ee^X$ and state our main assumptions.  
In Section \ref{sec:euro} we review and extend the results from \cite{rrvd} for pricing and replicating claims on $(X_T,\<X\>_T)$.  
These results will be needed for the barrier-style claims considered in Sections \ref{sec:ko}, \ref{sec:sbki} and \ref{sec:sbr}.
Section \ref{sec:ko} focuses on knock-out claims,
Section \ref{sec:sbki} examines knock-in claims
and Section \ref{sec:sbr} studies rebate claims.
Concluding remarks and directions for future research are offered in Section \ref{sec:conclusion}.

%
%

\section{Model and assumptions}
\label{sec:model}
We consider a frictionless market (i.e., no transaction costs) and fix an arbitrary but finite time horizon $T<\infty$.  For simplicity, we assume zero interest rates, no arbitrage, and take as given an equivalent martingale measure (EMM) $\Pb$ chosen by the market on a complete filtered probability space $(\Om,\Fc,\Fb,\Pb)$.  The filtration $\Fb=(\Fc_t)_{0 \leq t \leq T}$ represents the history of the market.  All stochastic processes defined below live on this probability space and all expectations are with respect to $\Pb$ unless otherwise stated.
\par
Let $B = (B_t)_{0 \leq t \leq T}$ represent the value of a zero-coupon bond maturing at time $T$.  As the risk-free rate of interest is zero by assumption, we have $B_t = 1$ for all $t \in [0,T]$.  Let $S = (S_t)_{0 \leq t \leq T}$ represent the value of a risky asset.  We assume $S$ is strictly positive and has continuous sample paths.  To rule out arbitrage, it is well-known that the asset $S$ must be a martingale under the pricing measure $\Pb$.  As such, there exists a non-negative, $\Fb$-adapted stochastic process $\sig = (\sig_t)_{0 \leq t \leq T}$ such that
\begin{align}
\dd S_t
	&=	\sig_t S_t \dd W_t , &
S_0
	&>	0 , 
\end{align} 
where $W$ is a Brownian motion with respect to the pricing measure $\Pb$ and the filtration $\Fb$.  Henceforth, the process $\sig$ will be referred to as the \emph{volatility process}.  
We assume that the volatility process $\sig$ is right-continuous and $\Fb$-adapted,
that it evolves independently of $W$ and that it satisfies
\begin{align}
\int_0^T \sig_t^2 \dd t
	&< c < \infty , \label{eq:bound}
\end{align}
for some arbitrarily large but finite constant $c>0$.   Note that $\sig$ may experience jumps and is not required to be Markovian.  
Define the log price process $X = (X_t)_{0 \leq t \leq T}$ by
\begin{align}
X_t
	&=	\log S_t .
\end{align}
As $S>0$, the process $X$ is well-defined and finite for all $t \in [0,T]$.  By It\^o's Lemma, 
\begin{align}
\dd X_t
	&=	-\tfrac{1}{2} \sig_t^2 \dd t + \sig_t \dd W_t , &
X_0
	&=	\log S_0 . \label{eq:dX}
\end{align}
Note that a claim on (the path of) $S$ can always be expressed as a claim on (the path of) $X = \log S$. 
\par
For any $\Fb$-stopping time $\tau$, define its $T$-bounded counterpart, the stopping time
\begin{align}
\tau^*
	&:= \tau \wedge T .
\end{align}
Let $C_{\tau^*}(K)$ denote the time $\tau^*$ price of a European call written on $S$ with maturity date $T$ and strike price $K > 0$, and let $P_{\tau^*}(K)$ denote the price of a European put written on $S$ with the same strike and maturity.  By no-arbitrage arguments, 
\begin{align}
C_{\tau^*}(K)
	&=	\Eb_{\tau^*} (S_T - K)^+  = \Eb_{\tau^*} (\ee^{X_T} - K)^+ , &
P_{\tau^*}(K)
	&=	\Eb_{\tau^*} (K - S_T)^+ = \Eb_{\tau^*}  (K - \ee^{X_T})^+ , \label{eq:call.put}
\end{align}
where we have introduced the shorthand notation $\Eb_{\tau^*} \, \cdot \, := \Eb[ \, \cdot \, | \Fc_{\tau^*}]$.
For convenience, we will sometimes refer to a European call or put written on $X$ rather than $S$ with the understanding that these are equivalent.
We assume that a European call or put with maturity $T$ trades at every strike $K \in (0,\infty)$.  As demonstrated by \cite{breeden}, this assumption is equivalent to knowing the distribution of $X_T$ under $\Pb$.  This assumption additionally guarantees, as \cite{carrmadan1998} show, that any $T$-maturity European claim on $X_T$ can perfectly hedged with a static portfolio of the bonds $B$, shares of the underlying $S$ and calls and puts.  Although in reality, calls and puts trade at only finitely many strikes, our results retain relevance; \cite{leung-lorig} show how to  adjust static hedges optimally when calls and puts are traded at only discrete strikes in a finite interval.

%
%

\section{European-style claims}
\label{sec:euro}
Under the assumptions of Section \ref{sec:model}, \cite{rrvd} show how to price and replicate the real and imaginary parts of a claim with a payoff of the form $\ee^{\ii \om X_T + \ii s \<X\>_T}$, where $\om,s \in \Cb$.  They then use these \emph{exponential claims} as building blocks to price and replicate more general claims with payoffs of the form $\varphi(X_T,\< X \>_T)$.  In this section, we briefly review the main results from \cite{rrvd} and derive some extensions needed in subsequent sections.
\par
Throughout this paper, we will distinguish between \emph{European} claims, which have path-\emph{independent} payoffs of the form $\varphi(X_T)$, and \emph{European-style} claims, which have path-\emph{dependent} payoffs of the form
\begin{align}
\text{European-style}:&&
\varphi(X_T,\<X\>_T) .
\end{align}
We use the phrase ``European-style'' to indicate that a claim payoff depends only on the terminal values $X_T$ and $\<X\>_T$ and not on any barrier event (e.g., knock-in or knock-out).


\subsection{Pricing and replicating power-exponential payoffs}
In what follows, we shall consider claims with $\Cb$-valued payoffs.  
The pricing and hedging results are understood to hold for the real and imaginary parts separately.
We begin by relating the characteristic function of $(X_T,\<X\>_T)$ to the characteristic function of $X_T$ only.

\begin{theorem}
\label{thm:E1=E2}
Let $\om,s \in \Cb$.  Define $u:\Cb^2 \to \Cb$ as either of the following
\begin{align}
u \equiv u_\pm(\om,s)
	&=	\ii \( - \tfrac{1}{2} \pm \sqrt{\tfrac{1}{4} - \om^2 - \ii \om + 2 \ii s} \) . \label{eq:u}
\end{align}
Then, for any $\Fb$-stopping time $\tau$, we have
\begin{align}
\Eb_{\tau^*} \ee^{\ii \om X_T + \ii s \<X\>_T} 
	&=	\ee^{\ii (\om-u) X_{\tau^*} + \ii s \<X\>_{\tau^*}}\Eb_{\tau^*} \ee^{\ii u X_T } . \label{eq:E1=E2}
\end{align}
\end{theorem}

\begin{proof}
A proof of Theorem \ref{thm:E1=E2} is given in \cite[Proposition 5.1]{rrvd}.  We repeat it here as the conditioning arguments below will be used in subsequent sections.  Let $\Fc_T^\sig$ denote the sigma-algebra generated by $(\sig_t)_{0 \leq t \leq T}$.  Then $(\<X\>_T - \<X\>_{\tau^*}) \in \Fc_{\tau^*} \vee \Fc_T^\sig$ and
\begin{align}
X_T - X_{\tau^*} | \Fc_{\tau^*} \vee \Fc_T^\sig
	&\sim \Nc( m , v^2 ) , &
m
	&=	-\tfrac{1}{2} ( \<X\>_T - \<X\>_{\tau^*} ) , &
v^2
	&=	\<X\>_T - \<X\>_{\tau^*} . \label{eq:X.normal}
\end{align}
Thus, by the characteristic function of a normal random variable
\begin{align}
Z
	&\sim \Nc( m , v^2 ) , &
\Eb \ee^{\ii \om Z}
	&=	\ee^{\ii m \om - \tfrac{1}{2} v^2 \om^2} , \label{eq:char.Z}
\end{align}
we have
\begin{align}
\Eb_{\tau^*} \ee^{ \ii \om (X_T - X_{\tau^*}) + \ii s (\<X\>_T - \<X\>_{\tau^*})} 
	&=	\Eb_{\tau^*} \ee^{\ii s (\<X\>_T - \<X\>_{\tau^*})} \Eb[ \ee^{ \ii \om (X_T - X_{\tau^*})}| \Fc_{\tau^*} \vee \Fc_T^\sig ] \\
	&=	\Eb_{\tau^*} \Eb[ \ee^{ (\ii s - ( \om^2 + \ii \om)/2 )(\<X\>_T - \<X\>_{\tau^*})} | \Fc_{\tau^*} \vee \Fc_T^\sig ] &
	&		\text{(by \eqref{eq:X.normal} and \eqref{eq:char.Z})}\\
	&=	\Eb_{\tau^*} \Eb[ \ee^{ ( - ( u^2 + \ii u)/2 )(\<X\>_T - \<X\>_{\tau^*})} | \Fc_{\tau^*} \vee \Fc_T^\sig ] &
	&		\text{(by \eqref{eq:u})}\\
	&=	\Eb_{\tau^*} \Eb[ \ee^{ \ii u (X_T - X_{\tau^*})}| \Fc_{\tau^*} \vee \Fc_T^\sig ] &
	&		\text{(by \eqref{eq:X.normal} and \eqref{eq:char.Z})} \\
	&=	\Eb_{\tau^*} \ee^{ \ii u (X_T - X_{\tau^*})}  . \label{eq:E1=E2.2}
\end{align}
Multiplying \eqref{eq:E1=E2.2} by $\ee^{\ii \om X_{\tau^*} + \ii s \<X\>_{\tau^*}}$ yields \eqref{eq:E1=E2}.
\end{proof}

\begin{corollary}
\label{cor:pow-exp}
Fix $\om,s \in \Cb$ and $n,m \in \{0\} \cup \Nb$. 
Assume $\frac{1}{4} - \ii \om + 2 \ii s - \om^2 \neq 0$.
Let $u : \Cb^2 \to \Cb$ be as defined in \eqref{eq:u}.  
Then
\begin{align}
\Eb_{\tau^*} X_T^n \<X\>_T^m \ee^{\ii \om X_T  + \ii s \<X\>_T }
	&=	\Eb_{\tau^*} \sum_{j=0}^n \sum_{k=0}^m \bigg( 
				\binom{n}{j} \binom{m}{k} 
				(-\ii \d_\om)^j (-\ii \d_s)^k \ee^{\ii (\om - u(\om,s))X_{\tau^*} + \ii s \<X\>_{\tau^*}} 
			\bigg) \\ & \quad 
			\times (-\ii \d_\om)^{n-j} (-\ii \d_s)^{m-k} \ee^{\ii u(\om,s) X_T} . \label{eq:name}
\end{align}
\end{corollary}

\begin{proof}
We have
\begin{align}
\Eb_{\tau^*} X_T^n \<X\>_T^m \ee^{\ii \om X_T  + \ii s \<X\>_T }
	&=	(-\ii \d_\om)^n (-\ii \d_s)^m \Eb_{\tau^*} \ee^{\ii \om X_T  + \ii s \<X\>_T } \\
	&=	(-\ii \d_\om)^n (-\ii \d_s)^m \ee^{\ii (\om - u(\om,s))X_{\tau^*} + \ii s \<X\>_{\tau^*}} \Eb_{\tau^*} \ee^{\ii u(\om,s) X_T} \\
	&=	\text{R.H.S. of \eqref{eq:name}} ,
\end{align}
where the first equality follows from the Leibniz integral rule, the second equality follows from Theorem \ref{thm:E1=E2}, and the last equality follows from the Leibniz integral rule and algebra.  The two applications of the Leibniz rule are justified as follows: for any $n,m \in \{ 0 \} \cup \Nb $ and $\om,s \in \Cb$ there exists a constant $c_1 > 0$ such that 
\begin{align}
| \d_\om^n \d_s^m \ee^{\ii \om x + \ii s v} | 
	&< c_1 \ee^{ c_1 (|x| + |v|)}, &
\Eb_0 c_1 \ee^{ c_1 (|X_T| + |\<X\>_T|)}
	&< \infty , \label{eq:leibniz}
\end{align}
where the finiteness of the expectation follows from \eqref{eq:bound}.
\end{proof}

\begin{remark}[Notation]
Throughout this manuscript, when it causes no confusion, we will omit the subscript $\pm$ and the arguments $(\om,s)$ from $u_\pm(\om,s)$ (and other functions/processes) in order to ease notation.
\end{remark}

We now recall a classical result from \cite{carrmadan1998}.  Suppose a function $f$ can be expressed as the difference of convex functions.  Then $f$ can be represented as a linear combination of call and put payoffs.  Specifically, for any $\kappa \in \Rb_+$ we have
\begin{align}
f(s)
	&=		f(\kappa) + f'(\kappa) \Big( (s - \kappa)^+ - (\kappa-s)^+ \Big)  
				+ \int_0^\kappa f''(K)(K-s)^+ \dd K   + \int_\kappa^\infty  f''(K)(s-K)^+\dd K . \label{eq:f}
\end{align}
Here, $f'$ is the left-derivative of $f$ and $f''$ is the second derivative, which exists as a generalized function.
Replacing $s$ in \eqref{eq:f} with the random variable $S_T$, choosing $\kappa = S_{\tau^*}$ and taking the $\Fc_{\tau^*}$-conditional expectation, one obtains
\begin{align}
\Eb_{\tau^*} f(S_T) 
	&=		f(S_{\tau^*}) B_{\tau^*} + \int_0^{S_{\tau^*}}   f''(K)P_{\tau^*}(K) \dd K + \int_{S_{\tau^*}}^\infty   f''(K) C_{\tau^*}(K) \dd K, \label{eq:f.call.put}
\end{align}
using $B_{\tau^*}=1$ and \eqref{eq:call.put}.  
Choosing $f$ so that $f(\ee^{X_T})$ is equal to the right-hand side of \eqref{eq:name} 
one obtains the price of a European-style power-exponential claim $\Eb_{\tau^*} X_T^n \<X\>_T^m \ee^{\ii \om X_T + \ii s \<X\>_T}$ in terms of (observable) European call and put prices.

Having priced European-style power-exponential claims relative to calls and puts, we turn to replication.  
Given a $\Cb^d$-valued price process $V = (V_t)_{0 \leq t \leq T}$, a $\Cb^d$-valued portfolio 
process $\Del = (\Del_t)_{0 \leq t \leq T}$ is said to be \emph{self-financing} if its value $\Pi$ satisfies
\begin{align}
\dd \Pi_t
	&=	\sum_{i=1}^d \Del_{t-}^i \dd V_t^i ,
&\text{where }
\Pi_t
	&:=	\sum_{i=1}^d \Del_t^i V_t^i  	
	 \label{eq:self.financing}
\end{align}
If both $\Del$ and $V$ are $\Rb^d$-valued, this definition corresponds to the usual notion of a self-financing portfolio.  The following theorem gives a self-financing replication strategy for European-style exponential claims.  
\begin{theorem}[Replication of European-style exponential claims]
\label{thm:complex}
Fix $\om,s, \in \Cb$ and define processes $N=(N_t)_{0 \leq t \leq T}$ and $Q=(Q_t)_{0 \leq t \leq T}$ by
\begin{align}
N_t \equiv N_t(\om,s)
	&:= \ee^{\ii(\om-u) X_t + \ii s \<X\>_t} , &
Q_t \equiv Q_t(\om,s)
	&:= \Eb_t \ee^{\ii u X_T} , \label{eq:NQ}
\end{align}
where $u \equiv u(\om,s)$ is as given in \eqref{eq:u}.  
Define $\Pi = (\Pi_t)_{0 \leq t \leq T}$ by
\begin{align}
\Pi_t \equiv \Pi_t(\om,s)
	&=	 N_t Q_t + \( \frac{ \ii (\om-u) N_t Q_{t-}}{S_t} \) S_t + \Big( - \ii (\om-u) N_t Q_{t-} \Big) B_t . \label{eq:Pi}
\end{align}
Then the portfolio $(N_t, \ii (\om-u) N_t Q_{t-}/S_t, - \ii (\om-u) N_t Q_{t-})$ of assets $(Q, S, B)$ is self-financing in the sense of \eqref{eq:self.financing}, and it has terminal value 
\begin{align}
\Pi_T
	&=	\ee^{\ii \om X_T + \ii s \<X\>_T } . \label{eq:Pi.T}
\end{align}
\end{theorem}
\begin{proof}
From \eqref{eq:Pi}, with $B_t = 1$, we have $\Pi_t = N_t Q_t$ at any time $t \in [0,T]$ .  
In particular, at the maturity date $T$, using \eqref{eq:NQ}, we have 
\begin{align}
\Pi_T 
	&=	  N_T Q_T 
	=		  \ee^{\ii ( \om - u) X_T + \ii s \<X\>_T} \Eb_T \ee^{\ii u X_T}
	=		  \ee^{\ii \om X_T + \ii s \<X\>_T } ,
\end{align}
which establishes \eqref{eq:Pi.T}.
To prove that $\Pi$ satisfies the self-financing condition \eqref{eq:self.financing}, observe that
\begin{align}
\Eb_t  \ee^{\ii \om X_T + \ii s \<X\>_T}
	&=	 \ee^{\ii(\om - u)X_t + \ii s \<X\>_t} \Eb_t \ee^{\ii u X_T}
	=		 N_t Q_t
	=		\Pi_t . \label{eq:martingale}
\end{align}
The left-hand side of \eqref{eq:martingale} is a martingale by iterated conditioning.  Thus, the process $\Pi$ must also be a martingale.  The process $Q$ must be a martingale by the same reasoning.  Next, 
\begin{align}
\dd \Pi_t
	&=	\dd (  N_t Q_t )
	=		 N_t \dd Q_t + Q_{t-} \dd N_t +  \dd[N,Q]_t \\
	&=	 N_t \dd Q_t + \( \frac{ \ii (\om-u) N_t Q_{t-}}{S_t} \) \dd S_t + \dd A_t ,
\end{align}
where $A = (A_t)_{0 \leq t \leq T}$ has finite variation.  As $\Pi$, $Q$ and $S$ are martingales, it follows that $A$ is a local martingale.  Moreover, as sample paths of $S$ are continuous, so too are the sample paths of $N$ and, hence, the sample paths of $A$.  As a finite variation, continuous local martingale, $A$ must be constant.  Thus $\dd A_t = 0$ and 
\begin{align}
\dd \Pi_t
	&=	N_t \dd Q_t + \( \frac{ \ii (\om-u) N_t Q_{t-}}{S_t} \) \dd S_t \\
	&=	N_t \dd Q_t + \( \frac{ \ii (\om-u) N_t Q_{t-}}{S_t} \) \dd S_t + \Big( -  \ii (\om-u) N_t Q_{t-} \Big) \dd B_t , \label{eq:d.Pi}
\end{align}
using $\dd B_t = 0$.  Comparing \eqref{eq:Pi} with \eqref{eq:d.Pi} establishes the self-financing condition.
\end{proof}

\begin{corollary}[Replication of European-style power-exponential claims]
Fix $\om,s, \in \Cb$ such that $2 \ii s-\omega ^2-\ii \om +\tfrac{1}{4} \neq 0$.
For any $n,m \in \{0\} \cup \Nb$ let the processes $N^{(n,m)}=(N_t^{(n,m)})_{0 \leq t \leq T}$ and $Q^{(n,m)}=(Q_t^{(n,m)})_{0 \leq t \leq T}$ be given by
\begin{align}
N_t^{(n,m)} \equiv N_t^{(n,m)}(\om,s)
	&:= (-\ii \d_\om)^n (-\ii\d_s)^m\ee^{\ii(\om-u) X_t + \ii s \<X\>_t} , \label{eq:N.nm} \\
Q_t^{(n,m)} \equiv Q_t^{(n,m)}(\om,s)
	&:= \Eb_t (-\ii \d_\om)^n (-\ii\d_s)^m \ee^{\ii u X_T} , \label{eq:Q.nm}
\end{align}
where $u \equiv u(\om,s)$ is as given in \eqref{eq:u}.  
Define the process $\Pi^{(n,m)}=(\Pi_t^{(n,m)})_{0 \leq t \leq T}$ by
\begin{align}
\Pi_t^{(n,m)} \equiv \Pi_t^{(n,m)}(\om,s)
	&=	 \sum_{j=0}^{n} \sum_{k=0}^{m} \binom{n}{j} \binom{m}{k} N_t^{(j,k)}  Q_t^{(n-j,m-k)}  \\ & \quad
				+ \Big( (-\ii \d_\om)^n (-\ii\d_s)^m  \frac{ \ii (\om-u) N_t Q_{t-}}{S_t} \Big) S_t \\ & \quad
				+ \Big( - (-\ii \d_\om)^n (-\ii\d_s)^m \ii (\om-u) N_t Q_{t-} \Big) B_t , \label{eq:Pi.2}
\end{align}
Then $\Pi^{(n,m)}$ is the value of a self-financing portfolio in the sense of \eqref{eq:self.financing} and satisfies
\begin{align}
\Pi_T^{(n,m)}
	&=	X_T^n \<X\>_T^m \ee^{\ii \om X_T + \ii s \<X\>_T} . \label{eq:Pi.T.2}
\end{align}
\end{corollary}

\begin{proof}
Throughout this proof, all uses of the Leibniz rule are justified by \eqref{eq:leibniz}.
From \eqref{eq:Pi.2}, at any time $t \in [0,T]$ we have
\begin{align}
\Pi_t^{(n,m)} 
	&=	\sum_{j=0}^{n} \sum_{k=0}^{m} \binom{n}{j} \binom{m}{k} N_t^{(j,k)}  Q_t^{(n-j,m-k)}
	= 	(-\ii \d_\om)^n (-\ii \d_s)^m N_t Q_t , \label{eq:step1}
\end{align}
using $B_t = 1$, equation \eqref{eq:N.nm} and equation \eqref{eq:Q.nm}.  In particular, at the maturity date $T$, we have 
\begin{align}
\Pi_T^{(n,m)} 
	&=	  (-\ii \d_\om)^n (-\ii \d_s)^m N_T Q_T 
	=		  (-\ii \d_\om)^n (-\ii \d_s)^m \ee^{\ii ( \om - u) X_T + \ii s \<X\>_T} \Eb_T \ee^{\ii u X_T} \\
	&=		(-\ii \d_\om)^n (-\ii \d_s)^m \ee^{\ii \om X_T + \ii s \<X\>_T } 
	=			X_T^n \<X\>_T^m \ee^{\ii \om X_T + \ii s \<X\>_T} ,
\end{align}
which establishes \eqref{eq:Pi.T.2}.
To prove that $\Pi^{(n,m)}$ satisfies the self-financing condition \eqref{eq:self.financing}, observe that
\begin{align}
\Eb_t  X_T^n \<X\>_T^m  \ee^{\ii \om X_T + \ii s \<X\>_T}
	&=	(-\ii \d_\om)^n (-\ii \d_s)^m \ee^{\ii(\om - u)X_t + \ii s \<X\>_t} \Eb_t \ee^{\ii u X_T} \\
	&=		(-\ii \d_\om)^n (-\ii \d_s)^m  N_t Q_t
	=		\Pi_t^{(n,m)}  . \label{eq:martingale.2}
\end{align}
The left-hand side of \eqref{eq:martingale.2} is a martingale by iterated conditioning, so $\Pi^{(n,m)}$ must also be a martingale.  For any $j,k \in \{0\} \cup \Nb$ the process $Q^{(j,k)}$ is a martingale by the same reasoning.  Next, 
by \eqref{eq:step1}, 
\begin{align}
\dd \Pi_t^{(n,m)}
	&=	\sum_{j=0}^n \sum_{k=0}^m \binom{n}{j} \binom{m}{k} 
			\Big( N_t^{(k,j)} \dd Q_t^{(n-j,m-k)} + Q_{t-}^{(n-j,m-k)} \dd N_t^{(j,k)} + \dd[N_t^{(j,k)},Q_t^{(n-j,m-k)}]_t \Big) \\
	&=	\sum_{j=0}^n \sum_{k=0}^m \binom{n}{j} \binom{m}{k}  N_t^{(j,k)} \dd Q_t^{(n-j,m-k)} 
			+ \Big( (-\ii \d_\om)^n (-\ii\d_s)^m  \frac{ \ii (\om-u) N_t Q_{t-}}{S_t} \Big) \dd S_t + \dd A_t^{(n,m)} ,
\end{align}
where $A^{(n,m)}=(A_t^{(n,m)})_{0 \leq t \leq T}$ has finite variation.  As $\Pi^{(n,m)}$, $S$, and $Q^{(j,k)}$ for any $j,k$ are martingales, it follows that $A^{(n,m)}$ is a local martingale.  Moreover, as sample paths of $S$ are continuous, so too are the sample paths of $N^{(j,k)}$ for any $j,k$ and, hence, the sample paths of $A^{(n,m)}$.  As a finite variation, continuous local martingale, $A^{(n,m)}$ must be constant.  Thus $\dd A_t^{(n,m)} = 0$ and 
\begin{align}
\dd \Pi_t^{(n,m)}
	&=	 \sum_{j=0}^{n} \sum_{k=0}^{m} \binom{n}{j} \binom{m}{k} N_t^{(j,k)}  \dd Q_t^{(n-j,m-k)}  
				+ \Big( (-\ii \d_\om)^n (-\ii\d_s)^m  \frac{ \ii (\om-u) N_t Q_{t-}}{S_t} \Big) \dd S_t \\
	&=	\sum_{j=0}^{n} \sum_{k=0}^{m} \binom{n}{j} \binom{m}{k} N_t^{(j,k)}  \dd Q_t^{(n-j,m-k)}  
				+ \Big( (-\ii \d_\om)^n (-\ii\d_s)^m  \frac{ \ii (\om-u) N_t Q_{t-}}{S_t} \Big) \dd S_t \\ & \quad
				+ \Big( - (-\ii \d_\om)^n (-\ii\d_s)^m \ii (\om-u) N_t Q_{t-} \Big) \dd B_t , \label{eq:d.Pi.2}
\end{align}
using $\dd B_t = 0$.  Comparing \eqref{eq:Pi.2} and \eqref{eq:d.Pi.2} establishes the self-financing condition.  
\end{proof}


\begin{example}[Sanity check: hedging a variance swap]
To replicate the floating leg of a variance swap which pays $\<X\>_T$, 
 take $(n,m)=(0,1)$ in \eqref{eq:Pi.2}, which yields
\begin{align}
\Pi_t^{(0,1)}
	&=	N_t^{(0,1)} Q_t + N_t Q_t^{(0,1)} \\ &\qquad
			+ \frac{1}{S_t} \Big( (\d_s u)N_t Q_{t-} + (\om-u) (\d_s N_t) Q_{t-}  + (\om-u) N_t (\d_s Q_{t-}) \Big) S_t \\ &\qquad
			- \Big( (\d_s u)N_t Q_{t-} + (\om-u) (\d_s N_t) Q_{t-}  + (\om-u) N_t (\d_s Q_{t-}) \Big) B_t .
\end{align}
In particular, for $u = u_+$ and $(\om,s)=(0,0)$, we have
\begin{align}
u(0,0)
	&=	0 , &
\d_s u(0,0)
	&=	-2 , &
-\ii \d_s N_t(0,0)
	&=	2 X_t + \<X\>_t , &
-\ii \d_s Q_t(0,0)
	&=	-2 \Eb_t X_T ,
\end{align}
and
\begin{align}
\Pi_t^{(0,1)}(0,0)
	&=	\Big( 2 X_t + \<X\>_t \Big) + \Eb_t (-2X_T) + \frac{2}{S_t} S_t - 2 B_t \\
	&=	- 2 \Eb_t (X_T-X_0) + \frac{2}{S_t} S_t + \Big( -2 + \<X\>_t + 2 X_t - 2 X_0 \Big) B_t \\
	&=	- 2 \Eb_t \log \Big( \frac{S_T}{S_0} \Big) + \frac{2}{S_t} S_t + \Big( -2 + \int_0^t \frac{2}{S_r}\dd S_r \Big) B_t, 
\end{align}
which recovers the classical hedging strategy for a variance swap: hold $-2$ European $\log$ contracts, keep two units of currency in $S$ at all times $t \in [0,T]$ and finance the position with zero-coupon bonds.
\end{example}


\subsection{Pricing and replicating more general payoffs}
As previously mentioned, \cite{rrvd} use complex exponential claims as building blocks to construct prices and replication strategies for a variety of other more complicated claims, including claims that pay $\<X\>_T^r$ where $-\infty<r<1$ (see \cite[Propositions 7.1 and 7.2]{rrvd}).  For options on $\<X\>_T$ only, this is typically done via Laplace transforms.  For options on $(X_T,\<X\>_T)$ it will be helpful to introduce the \emph{generalized Fourier transform} $\Fv$ and \emph{inverse transform} $\Fv^{-1}$.   
For any functions $f : \Rb \to \Cb$ and $\fh : \Cb \to \Cb$ such that the following integrals exist, 
define
\begin{align}
\text{Fourier Transform}:&&
&\Fv[f](\om)
	:=	\frac{1}{2\pi}\int_{\Rb}  f(x) \ee^{- \ii \om x} \dd x, & 
\om
	&\in \Cb , \label{eq:FT}\\
\text{Inverse Transform}:&&
&\Fv^{-1}[\fh](x)
	:=	\int_{\Rb}  \fh(\om) \ee^{ \ii \om x} \dd \om_r, &
\om_r
	&=	\Re(\om) , \label{eq:iFT}
\end{align}
Consider now a European-style claim with a payoff of the form 
\begin{align}
\varphi(X_T,\<X\>_T) 
	&= f(X_T)\<X\>_T^m \ee^{ \ii s \<X\>_T} , &
f
	&:	\Rb \to \Cb , &
m
	&\in \{0\} \cup \Nb , &
s
	&\in \Cb . \label{eq:euro-style}
\end{align}
If $f = \Fv^{-1}[ \fh ]$ where $\fh = \Fv[ f ]$, then formally, we have
\begin{align}
\Eb_{\tau^*} \varphi(X_T,\<X\>_T) 
	&=	\Eb_{\tau^*} f(X_T)\<X\>_T^m \ee^{\ii s \<X\>_T} \\
	&=	\Eb_{\tau^*} f(X_T) (-\ii \d_s)^m  \ee^{ \ii s \<X\>_T} \\
	&=	\int_\Rb \fh(\om) (-\ii \d_s)^m  \Eb_{\tau^*} \ee^{\ii \om X_T + \ii s \<X\>_T}\dd \om_r  &
	&		\text{(as $f = \Fv^{-1}[\fh])$} \\
	&=	\int_\Rb  \fh(\om) (-\ii \d_s)^m  \ee^{\ii (\om-u(\om,s)) X_{\tau^*} + \ii s \<X\>_{\tau^*}}\Eb_{\tau^*} \ee^{\ii u(\om,s) X_T }\dd \om_r , &
	&		\text{(by Theorem \ref{thm:E1=E2})} \\
	&=	\Eb_{\tau^*} g(X_T,X_{\tau^*},\<X\>_{\tau^*}) , \label{eq:pricing} \\
g(X_T,X_{\tau^*},\<X\>_{\tau^*})
	&:=	\int_\Rb  \fh(\om) (-\ii \d_s)^m  \ee^{\ii (\om-u(\om,s)) X_{\tau^*} + \ii s \<X\>_{\tau^*}} \ee^{\ii u(\om,s) X_T }\dd \om_r. \label{eq:euro.g}
\end{align}
Assuming the various applications of Fubini and the Leibniz integral rule are justified, equation \eqref{eq:pricing} relates the value of a European-style claim with a payoff of the form \eqref{eq:euro-style} to the value of a European claim with payoff \eqref{eq:euro.g}.  Moreover, as
\begin{align}
\varphi(X_T,\<X\>_T)
	&=	\int_\Rb \fh(\om) \<X\>_T^m \ee^{\ii \om X_T + \ii s \<X\>_T} \dd \om_r ,
\end{align}
a replicating strategy for $\varphi(X_T,\<X\>_T)$ can be obtained by taking a (continuous) linear combination of replicating strategies for power-exponential claims with payoffs of the form $\<X\>_T^m \ee^{\ii \om X_T + \ii s \<X\>_T}$.

\section{Knock-out claims}
\label{sec:ko}
For any $H \in \Rb$ define the \emph{first hitting time to level $H$} as
\begin{align}
\tau_H
	&:=	\inf \{ t \geq 0 : X_t = H \} , &
H
	&\in \Rb , \label{eq:tau.L}
\end{align}
where $\inf\emptyset  := \infty$.  Next, for any $L,U \in \Rb$ with $L < X_0 < U$, define the \emph{first hitting time to level $L$ or $U$} as
\begin{align}
\tau_{L,U}
	&:=	\tau_L \wedge \tau_U , &
L
	&<  X_0 < U .  \label{eq:tau.LU}
\end{align}
Observe that $\tau_H$ and $\tau_{L,U}$ are $\Fb$-stopping times as are their $T$-bounded counterparts $\tau_H^*$ and $\tau_{L,U}^*$.


\subsection{Single barrier knock-out claims}
\label{sec:sbko}
This section considers \emph{single barrier knock-out} claims with payoffs of the form
\begin{align}
\text{Single barrier knock-out}:&&
	 \Ib{\tau_H > T}\varphi(X_T,\<X\>_T) . \label{eq:sbko.H}
\end{align}
The following strategy replicates a single barrier knock-out claim with a down-barrier $L < X_0$.

\begin{theorem}[Replication of single barrier knock-out claims]
\label{thm:sbko}
Fix $L < X_0$.  The following trading strategy replicates a single barrier knock-out claim with payoff 
\begin{align}
	 \Ib{\tau_L > T}\varphi(X_T,\<X\>_T) . \label{eq:sbko}
\end{align}
At time $0$ hold a European-style claim with payoff
\begin{align}
\varphi_L^\ko(X_T,\<X\>_T)
	&:=	 \Ib{X_T > L} \varphi(X_T,\<X\>_T) - \Ib{X_T < L} \ee^{X_T - L} \varphi(2 L - X_T,\<X\>_T) . \label{eq:phi.sbko}
\end{align}
If and when the claim knocks out, close the position in $\varphi_L^\ko(X_T,\<X\>_T)$ at no cost.
\end{theorem}

\begin{proof}
If $\tau_L > T$, then $X_T > L$ and thus, both the knock-out claim \eqref{eq:sbko} and the European-style claim \eqref{eq:phi.sbko} pay $\varphi(X_T,\<X\>_T)$.  
It remains to show that, when $\tau_L \leq T$, the European-style claim \eqref{eq:phi.sbko} has zero value at time $\tau_L$.
First, we note from \cite[Definition 2.6]{pcs} that $S=\ee^X$ satisfies \emph{geometric put-call symmetry}.  Thus, \cite[Theorem 5.3]{pcs} implies that
\begin{align}
\Eb_{\tau^*} G(X_T)
	&=	\Eb_{\tau^*} \ee^{X_T - X_{\tau^*}} G(2 X_{\tau^*} - X_T) . \label{eq:pcs}
\end{align} 
for any $\Fb$-stopping time $\tau$ and $G : \Rb \to \Cb$.  Therefore
\begin{align}
\Eb_{\tau^*} G(X_T,\<X\>_T)
	&=	\Eb_{\tau^*} \Eb[ G(X_T,\<X\>_T) | \Fc_{\tau^*} \vee \Fc_T^\sig ] \\
	&=	\Eb_{\tau^*} \Eb[ \ee^{X_T - X_{\tau^*}} G(2 X_{\tau^*} - X_T,\<X\>_T) | \Fc_{\tau^*} \vee \Fc_T^\sig ] \\
	&=	\Eb_{\tau^*} \ee^{X_T - X_{\tau^*}} G(2 X_{\tau^*} - X_T,\<X\>_T) , \label{eq:pcs2}
\end{align}
where the second equality follows from \eqref{eq:pcs} and the fact that $S=\ee^X$, conditioned on the path of $\sig$, satisfies geometric put-call symmetry.  Using \eqref{eq:pcs2} with $G(x,v) = \Ib{x>L} \varphi(x,v)$ and recalling that $\Ib{\tau_L \leq T} (X_{\tau_L^*} - L )=0$, we have $\Ib{\tau_L \leq T}\Eb_{\tau_L^*} \varphi_L^\ko(X_T,\<X\>_T) = 0$.
\end{proof}

\begin{remark}
For the single barrier knock-out claim $\Ib{\tau_U > T}\varphi(X_T,\<X\>_T)$ with \emph{up}-barrier $U > X_0$, the replication strategy is to hold at time 0 a European-style claim with payoff
\begin{align}
\varphi_U^\ko(X_T,\<X\>_T)
	&:=	 \Ib{X_T < U} \varphi(X_T,\<X\>_T) - \Ib{X_T > U} \ee^{X_T - U}\varphi(2U - X_T,\<X\>_T) , 
\end{align}
and clear the position at no cost if and when the barrier $U$ is hit.
\end{remark}

\begin{proposition}[Price of a single barrier knock-out power-exponential claim]
\label{prp:sbko-pow-exp}
Assume the distribution of $X_T$ has no point masses (a sufficient condition is that $\int_0^T \sig_t^2 \dd t > \eps > 0$).  
Then for any $L < X_0$, $j,k \in \{0\} \cup \Nb$ and $p,s \in \Cb$ we have
\begin{align}
\Eb \Ib{\tau_L > T} X_T^j \< X \>_T^k \ee^{\ii p X_T + \ii s \< X \>_T }
	&=	\lim_{n \to \infty} \Eb \Big( g_n(X_T) - h_n(X_T) \Big) , \label{eq:sbko-main}
\end{align}
where the functions $g_n$ and $h_n$ are given by
\begin{align}
g_n(X_T)
	&=	\int_\Rb  (- \ii \d_p)^j (- \ii \d_s)^k \Hh_n(\om - p) \ee^{- \ii (\om - p) L + \ii(\om - u(\om,s))X_0} \ee^{\ii u(\om,s) X_T} \dd \om_r, \\
h_n(X_T) 
	&=	\int_\Rb  (- \ii \d_p)^j (- \ii \d_s)^k \Hh_n(-\ii - \om - p) \ee^{- \ii (\om - p) L + \ii(\om - u(\om,s))X_0} \ee^{\ii u(\om,s) X_T}\dd \om_r, \\
\Hh_n(\om)
	&=	\frac{-\ii}{4n} \textup{csch} \Big( \frac{\pi \om}{2 n} \Big) . \label{eq:Hhat}
\end{align}
Here, the contour of integration in $g_n$ is chosen so that $-2n + p_i < \om_i < p_i$ and $2 \ii s-\omega ^2-\ii \om +\tfrac{1}{4} \neq 0$, and
the contour of integration in $h_n$ is chosen so that $-1 - p_i < \om_i < 2n - 1 - p_i$ and $2 \ii s-\omega ^2-\ii \om +\tfrac{1}{4} \neq 0$.
\end{proposition}

\begin{proof}
Let $H$ denote the Heaviside function and let $H_n$ ($n \in \Nb$) denote a smooth approximation of $H$.  Specifically, let
\begin{align}
H(x)
	&:=	\tfrac{1}{2}(1 + \text{sgn} \, x) , &
H_n(x) 
	&:= \tfrac{1}{2} ( 1 + \tanh n x ) . \label{eq:heaviside}
\end{align}
Observe that $H_n \to H$ pointwise as $n \to \infty$.  Now, as the price of any claim is equal to the price of its replicating portfolio, we have by Theorem \ref{thm:sbko} that
\begin{align}
&\Eb \Ib{\tau_L > T} X_T^j \<X\>_T^k \ee^{\ii p X_T + \ii s \<X\>_T} \\
	&=	\Eb \Big( \Ib{X_T > L} X_T^j \<X\>_T^k   \ee^{\ii p X_T + \ii s \<X\>_T} 
			- \Ib{X_T < L} (2L-X_T)^j \<X\>_T^k \ee^{X_T - L + \ii p (2 L - X_T) + \ii s \<X\>_T} \Big) \\
	&=	\Eb  \lim_{n \to \infty} \Big( H_n(X_T-L) X_T^j \<X\>_T^k  \ee^{\ii p X_T + \ii s \<X\>_T} 
			- H_n(L-X_T) (2L-X_T)^j \<X\>_T^k \ee^{X_T - L + \ii p (2 L - X_T) + \ii s \<X\>_T} \Big) \\
	&=	\lim_{n \to \infty} \Eb \Big( H_n(X_T-L) X_T^j \<X\>_T^k  \ee^{\ii p X_T + \ii s \<X\>_T} 
			- H_n(L-X_T) (2L-X_T)^j \<X\>_T^k \ee^{X_T - L + \ii p (2 L - X_T) + \ii s \<X\>_T} \Big) \\
	&=	\lim_{n \to \infty} (-\ii \d_p)^j(-\ii \d_s)^k \Eb \Big( H_n(X_T-L) \ee^{\ii p X_T + \ii s \<X\>_T} 
			- H_n(L-X_T) \ee^{X_T - L + \ii p (2 L - X_T) + \ii s \<X\>_T} \Big) , \label{eq:sbko-1}
\end{align}
where the second equality holds by absence of point masses, the third equality holds by Lebesgue's dominated convergence theorem and the last equality follows from the Leibniz integral rule.  Noting that $\Fv[H_n] = \Hh_n$ where $\Hh_n(\om)$ is defined for $-2n < \om_i < 0$, it follows that
\begin{align} 
&(-\ii \d_p)^j(-\ii \d_s)^k  \Eb H_n(X_T-L) \ee^{\ii p X_T + \ii s \<X\>_T}  \\
	&=	(-\ii \d_p)^j(-\ii \d_s)^k  \Eb \int_\Rb  \Hh_n(\om - p) \ee^{- \ii (\om - p) L } \ee^{\ii \om X_T + \ii s \<X\>_T} \dd \om_r&
	&		(-2n + p_i < \om_i < p_i) \\
	&=	(-\ii \d_p)^j(-\ii \d_s)^k  \int_\Rb \Hh_n(\om - p) \ee^{- \ii (\om - p) L } \Eb \ee^{\ii \om X_T + \ii s \<X\>_T} \dd \om_r &
	&		\text{(by Fubini)}\\
	&=	(-\ii \d_p)^j(-\ii \d_s)^k  \int_\Rb \Hh_n(\om - p) \ee^{ \ii (\om - p) L + \ii (\om - u(\om,s))X_0} \Eb \ee^{\ii u(\om,s) X_T}  \dd \om_r &
	&		\text{(by \eqref{eq:E1=E2})} \\
	&=	(-\ii \d_p)^j(-\ii \d_s)^k  \Eb \int_\Rb \Hh_n(\om - p) \ee^{ \ii (\om - p) L + \ii (\om - u(\om,s))X_0} \ee^{\ii u(\om,s) X_T} \dd \om_r  &
	&		\text{(by Fubini)} \\
	&=	\Eb g_n(X_T) , &
	&		\text{(by Leibniz)} \label{eq:sbko-2}
\end{align}
where the applications of Fubini twice and the Leibniz rule are justified as
$|\d_p^j \Hh_n(\om-p)| = \Oc( \ee^{ - |\om_r|/n })$
and
\begin{align} 
\Eb |\d_p^j \d_s^k \ee^{- \ii (\om - p) L + \ii \om X_T + \ii s \<X\>_T}| 
	&= \Oc(1)  , &
\Eb |\d_p^j \d_s^k  \ee^{\ii (\om - p) L + \ii (\om - u(\om,s))X_0 +\ii u(\om,s) X_T}| 
	&= \Oc(1) ,
\end{align}
as $|\om_r| \to \infty$ and as the contour of integration is chosen to avoid any singularities in the integrand.  Similarly, 
\begin{align}
(-\ii \d_p)^j(-\ii \d_s)^k \Eb H_n(L-X_T) \ee^{X_T - L + \ii p (2 L - X_T) + \ii s \<X\>_T}
	&=	\Eb h_n(X_T) . \label{eq:sbko-3}
\end{align}
Equation \eqref{eq:sbko-main} follows from \eqref{eq:sbko-1}, \eqref{eq:sbko-2} and \eqref{eq:sbko-3}.
\end{proof}

\begin{remark}
\label{rmk:heaviside}
The reason we must replace the Heaviside function $H$ by $\lim_{n \to \infty} H_n$ in Proposition \ref{prp:sbko-pow-exp} is that the Heaviside function's Fourier transform $\Hh(\om) = -\ii / (2 \pi \om)$ $(\om_i > 0)$ does not decay fast enough as $|\om_r| \to \infty$ to justify the second use of Fubini in \eqref{eq:sbko-2}. 
\end{remark}

\begin{remark}
Equation \eqref{eq:sbko-main} is an equation of the form
\begin{align}
\Eb F[X]
	&=	\lim_{n \to \infty} \Eb g_n(X_T) , \label{eq:functional}
\end{align}
where $F$ is a \emph{functional} of $X=(X_t)_{0 \leq t \leq T}$.  Observe that $\Eb F[X]$ is the price of a path-dependent claim and $\Eb g_n(X_T)$ is the price of a European (i.e., path-independent) claim.  Thus, we say that, in the limit as $n \to \infty$, the expected payoff $\Eb g_n$ \emph{prices} the claim $F[X]$. 
\end{remark}

Figure \ref{fig:sbko} plots the function whose expectation, in the limit as $n \to \infty$, prices a single barrier knock-out variance swap, which pays $\Ib{\tau_L >T} \<X\>_T$.


\subsection{Double barrier knock-out claims}
\label{sec:dbko}
This section considers \emph{double barrier knock-out} claims with payoffs of the form
\begin{align}
\text{Double barrier knock-out claim}:&&
\Ib{\tau_{L,U}>T} \varphi(X_T,\<X\>_T) . \label{eq:dbko}
\end{align}
The following theorem gives a replication strategy for such claims.

\begin{theorem}[Replication of double barrier knock-out claims]
\label{thm:dbko}
Suppose $L < X_0 < U$.
Let $\varphi: (L,U) \times \Rb_+ \to \Cb$ be bounded.
The following trading strategy replicates a double barrier knock-out claim with payoff \eqref{eq:dbko}.
At time $0$ hold a European-style claim with payoff
\begin{align}
\varphi_{L,U}^\ko(X_T,\<X\>_T)
	&:=	\sum_{n=-\infty}^\infty \ee^{-n\Del} \Big(
			\varphi^*({2n\Del+X_T},\<X\>_T) 
			- \ee^{X_T-L}\varphi^*({2n\Del+2L-X_T},\<X\>_T)
			\Big) , \label{eq:phi.dbko} \\
\varphi^*(X_T,\<X\>_T)
	&:=	\varphi(X_T,\<X\>_T)\Ib{L<X_T<U} , \label{eq:phi.star}
\end{align}
where $\Del := U-L$.  If and when the claim knocks out, clear the position in $\varphi_{L,U}^\ko(X_T,\<X\>_T)$ at no cost.
\end{theorem}

\begin{proof}
If $\tau_{L,U} > T$, then $L < X_T < U$ and thus, both the knock-out claim \eqref{eq:dbko} and the European-style claim \eqref{eq:phi.dbko} pay $\varphi(X_T,\<X\>_T)$.   It remains to show that, if $\tau_{L,U} \leq T$, the European-style claim \eqref{eq:phi.dbko} has zero value at time $\tau_{L,U}$.  Recalling once again that $S = \ee^X$ satisfies geometric put-call symmetry, we have by \cite[Theorem 5.18]{pcs} that
\begin{align}
\Ib{\tau_{L,U}\leq T} \Eb_{\tau_{L,U}^*} \varphi_{L,U}^\ko (X_T,v)
	&=	0 ,
\end{align}
which holds for any fixed $v \in \Rb_+$.  Thus
\begin{align}
\Ib{\tau_{L,U}\leq T} \Eb_{\tau_{L,U}^*} \varphi_{L,U}^\ko (X_T,\<X\>_T)
	&=	\Ib{\tau_{L,U}\leq T} \Eb_{\tau_{L,U}^*} \Eb[ \varphi_{L,U}^\ko (X_T,\<X\>_T) | \Fc_{\tau_{L,U}^*} \vee \Fc_T^\sig]
	=	0 ,
\end{align}
because $\<X\>_T \in \Fc_T^\sig$ and the process $S=\ee^X$, conditioned on the path of $\sig$, satisfies geometric put-call symmetry.
\end{proof}

\begin{proposition}[Prices of double barrier knock-out power-exponential claims]
\label{prp:dbko-pow-exp}
Assume the distribution of $X_T$ has no point masses (a sufficient condition is that $\int_0^T \sig_t^2 \dd t > \eps > 0$).  
Then for any $L < X_0 < U$, $j,k \in \{0\} \cup \Nb$ and $p,s \in \Cb$ we have
\begin{align}
\Eb \Ib{\tau_{L,U}>T} X_T^j \<X\>_T^k \ee^{\ii p X_T + \ii s \<X\>_T}
	&=	\lim_{q \to \infty} \lim_{m \to \infty} \Eb \sum_{n=-q}^q \ee^{-n\Del} \Big( g_{n,m}(X_T) - h_{n,m}(X_T) \Big) , \label{eq:dbko-pow-exp}
\end{align}
where the functions $g_{n,m}$ and $h_{n,m}$ are given by
\begin{align}
g_{n,m}(X_T)
	&=	\int_\Rb \ee^{\ii \om 2 n \Del}  (-\ii \d_p)^j (-\ii \d_s)^k 
			\Big( \ee^{- \ii (\om - p) L}  - \ee^{- \ii (\om - p) U} \Big) \\ &\quad \quad
			\times \Hh_m(\om-p) \ee^{\ii (\om - u(\om,s))X_0+\ii u(\om,s) X_T} \dd \om_r, \\
h_{n,m}(X_T)
	&=	\int_\Rb \ee^{(1-\ii \om)( 2 n \Del- 2 L)+L} (-\ii \d_p)^j (-\ii \d_s)^k
			\Big( \ee^{- \ii (-\ii - p - \om) L}  - \ee^{- \ii ( -\ii - p -\om ) U} \Big) \\ &\quad \quad
			\times\Hh_m(-\ii - p - \om) \ee^{\ii (\om - u(\om,s))X_0+\ii u(\om,s) X_T} \dd \om_r,
\end{align}
with $\Hh_m$ as defined in \eqref{eq:Hhat}.  The contour of integration for $g_{n,m}$ must be chosen so that $-2m + p_i < \om_i < p_i$ and $2 \ii s-\omega ^2-\ii \om +\tfrac{1}{4} \neq 0$ and the contour of integration for $h_{n,m}$ must be chosen so that $-1 - p_i < \om_i < 2m - 1 - p_i$ and $2 \ii s-\omega ^2-\ii \om +\tfrac{1}{4} \neq 0$.
\end{proposition}

\begin{proof}
As many of the arguments for passing limits and derivatives through integrals and expectations are analogous to those given in the proof of Proposition \ref{prp:sbko-pow-exp}, we shall not repeat them here.
Noting that the value of any claim is equal to the value of its replication portfolio, we have from Theorem \ref{thm:dbko} that 
\begin{align}
\Eb \Ib{\tau_{L,U}>T}\varphi(X_T,\<X\>_T)
	&=	\sum_{n=-\infty}^\infty \ee^{-n\Del} \Eb \Big(
			\varphi^*({2n\Del+X_T},\<X\>_T) - \ee^{X_T-L} \varphi^*({2n\Del+2L-X_T},\<X\>_T)
			\Big) , \label{eq:dbko-sum}
\end{align}
where passing the expectation through the infinite sum is allowed by the arguments given in the proof of \cite[Theorem 5.18]{pcs}.

Examining the first term in the expectation above, with $\varphi(X_T,\<X\>_T) = X_T^j \<X\>_T^k \ee^{\ii p X_T + \ii s \<X\>_T}$, we have
\begin{align}
&\Eb \varphi^*({2n\Del+X_T},\<X\>_T) \\
	&=	\Eb \Ib{L < 2n\Del + X_T < U} (2n\Del+X_T)^j \<X\>_T^k \ee^{\ii p (2n\Del+X_T) + \ii s \<X\>_T} \\
	&=	\lim_{m \to \infty} \Eb \Big( H_m(2n\Del+X_T - L) - H_m(2n\Del+X_T - U) \Big) (2n\Del+X_T)^j \<X\>_T^k \ee^{\ii p (2n\Del+X_T) + \ii s \<X\>_T} \\
	&=	\lim_{m \to \infty} (-\ii \d_p)^j (-\ii \d_s)^k \Eb \Big( H_m( 2n\Del+X_T - L) - H_m( 2n\Del+X_T - U) \Big) \ee^{\ii p (2n\Del+X_T) + \ii s \<X\>_T} ,
\end{align}
where $H_m(x) := \tfrac{1}{2} ( 1 + \tanh m x )$.
Next, using $\Fv[ H_m ] = \Hh_m$, we compute
\begin{align}
&(-\ii \d_p)^j (-\ii \d_s)^k \Eb \Big( H_m( 2n\Del+X_T - L) - H_m( 2n\Del+X_T - U) \Big) \ee^{\ii p (2n\Del+X_T) + \ii s \<X\>_T} \\
	&=	\int_\Rb \ee^{\ii \om 2 n \Del}  (-\ii \d_p)^j (-\ii \d_s)^k \Big( \ee^{- \ii (\om - p) L}  - \ee^{- \ii (\om - p) U} \Big) \Hh_m(\om-p) 
			\Eb \ee^{\ii \om X_T + \ii s \<X\>_T}\dd \om_r \\
	&=	\int_\Rb \ee^{\ii \om 2 n \Del}  (-\ii \d_p)^j (-\ii \d_s)^k \Big( \ee^{- \ii (\om - p) L}  - \ee^{- \ii (\om - p) U} \Big) \Hh_m(\om-p) 
			\ee^{\ii (\om - u(\om,s))X_0}\Eb \ee^{\ii u(\om,s) X_T}  \dd \om_r\\
	&=	\Eb g_{n,m}(X_T) .
\end{align}
Similarly, a straightforward computation shows
\begin{align}
\Eb \ee^{X_T-L} \varphi^*({2n\Del+2L-X_T},\<X\>_T)
	&=	\lim_{m \to \infty} \Eb h_{n,m}(X_T) .
\end{align}
Thus
\begin{align}
\Eb \Ib{\tau_{L,U}>T}\varphi(X_T,\<X\>_T)
	&=	\sum_{n=-\infty}^\infty \lim_{m \to \infty}  \ee^{-n\Del} \Eb \Big( g_{n,m}(X_T) - h_{n,m}(X_T) \Big) \\
	&=	\lim_{q \to \infty} \lim_{m \to \infty} \Eb \sum_{n=-q}^q \ee^{-n\Del} \Big( g_{n,m}(X_T) - h_{n,m}(X_T) \Big) ,
\end{align}
as claimed.
\end{proof}

Figure \ref{fig:dbko} plots the function whose expectation, in the limit as $m,q \to \infty$, prices a double barrier knock-out variance swap, which pays $\Ib{\tau_{L,U} >T} \<X\>_T$.

%
%



\section{Single barrier knock-in claims}
\label{sec:sbki}
This section considers \emph{single barrier knock-in claims} with payoffs of the form 
\begin{align}
\text{Single barrier knock-in}:&&
	&\Ib{\tau_H \leq T} \varphi(X_T-X_{\tau_H^*},\<X\>_T - \<X\>_{\tau_H^*}) .
\end{align}
The following trading strategy replicates a knock-in power-exponential claim with a single barrier $L<X_0$.

\begin{theorem}[Replication of single barrier knock-in power-exponential claims]
\label{thm:sbki}
Fix $L < X_0$, $n,m \in \{0\} \cup \Nb$ and $\om,s \in \Cb$ and assume $2 \ii s-\omega ^2-\ii \om +\tfrac{1}{4} \neq 0$.
The following trading strategy replicates the single barrier knock-in power-exponential payoff
\begin{align}
\Ib{\tau_L \leq T} (X_T-X_{\tau_L^*})^n (\<X\>_T - \<X\>_{\tau_L^*})^m \ee^{\ii \om (X_T-X_{\tau_L^*}) + \ii s ( \<X\>_T - \<X\>_{\tau_L^*}) } . \label{eq:sbki}
\end{align}
At time $0$ hold a European claim with payoff
\begin{align}
	(-\ii \d_\om)^n (-\ii \d_s)^m \psi_L^{\ki}(X_T;\om,s) , \label{eq:psi.sbki.0}
\end{align}
where we have defined
\begin{align}
\psi_L^{\ki}(X_T) \equiv \psi_L^{\ki}(X_T;\om,s)
	&=	\Ib{X_T < L} \ee^{(1 - \ii u) (X_T - L)} + \Ib{X_T \leq L}  \ee^{\ii u (X_T - L)} , \label{eq:psi.sbki}
\end{align}
with $u \equiv u_\pm(\om,s)$ as given in \eqref{eq:u}.  If and when the claim knocks in, exchange the claim \eqref{eq:psi.sbki.0} for the knock-in claim \eqref{eq:sbki} at no cost.  After the exchange, the knock-in claim \eqref{eq:sbki} can be replicated as a European-style power-exponential claim.
\end{theorem}

\begin{proof}
If $\tau_L > T$, then $X_T > L$ and thus both the knock-in claim \eqref{eq:sbki} and the European claim \eqref{eq:psi.sbki.0} expire worthless.  
It remains to show that, if $\tau_L \leq T$, the claim \eqref{eq:psi.sbki.0} can be exchanged for the claim \eqref{eq:sbki} at no cost.  Recalling that $S=\ee^X$ satisfies geometric put-call symmetry, we have from \cite[equation (5.7)]{pcs} that
\begin{align}
\Ib{\tau_L \leq T} \Eb_{\tau_L^*} \ee^{\ii u (X_T - L)} 
	&=	\Ib{\tau_L \leq T} \Eb_{\tau_L^*} \psi_L^{\ki}(X_T) . \label{eq:pcs.result}
\end{align}
Thus 
\begin{align}
&\Ib{\tau_L \leq T}  \Eb_{\tau_L^*} (X_T-X_{\tau_L^*})^n (\<X\>_T - \<X\>_{\tau_L^*})^m \ee^{\ii \om (X_T-X_{\tau_L^*}) + \ii s (\<X\>_T - \<X\>_{\tau_L^*}) } \\
	&=	\Ib{\tau_L \leq T} (-\ii \d_\om)^n (-\ii \d_s)^m \Eb_{\tau_L^*} \ee^{\ii \om (X_T-X_{\tau_L^*}) + \ii s (\<X\>_T - \<X\>_{\tau_L^*}) } &
	&		\text{(by Leibniz)} \\
	&=	\Ib{\tau_L \leq T} (-\ii \d_\om)^n (-\ii \d_s)^m \Eb_{\tau_L^*} \ee^{\ii u (X_T - L)}  &
	&		\text{(by \eqref{eq:E1=E2})} \\
	&=	\Ib{\tau_L \leq T} (-\ii \d_\om)^n (-\ii \d_s)^m \Eb_{\tau_L^*} \psi_L^{\ki}(X_T)   &
	&		\text{(by \eqref{eq:pcs.result})} \\
	&=	\Ib{\tau_L \leq T} \Eb_{\tau_L^*} (-\ii \d_\om)^n (-\ii \d_s)^m \psi_L^{\ki}(X_T) ,  &
	&		\text{(by Leibniz)} \label{eq:proof.sbki}
\end{align}
where the two uses of the Leibniz rule are justified by \eqref{eq:leibniz}.
\end{proof}

\begin{remark}
\label{rmk:sbki}
To replicate the single barrier knock-in power-exponential claim with payoff
\begin{align}
\Ib{\tau_U \leq T} (X_T-X_{\tau_U^*})^n (\<X\>_T - \<X\>_{\tau_U^*})^m \ee^{\ii \om (X_T - X_{\tau_U^*})+ \ii s ( \<X\>_T - \<X\>_{\tau_U^*}) }
\end{align}
where $U > X_0$, one should hold at time 0 a the European claim with payoff $ (-\ii \d_\om)^n (-\ii \d_s)^m\psi_U^{\ki}(X_T;\om,s)$ where
\begin{align}
\psi_U^{\ki}(X_T) \equiv \psi_U^{\ki}(X_T;\om,s)
	&=	 \Ib{X_T > U} \ee^{(1 - \ii u) (X_T - U)} + \Ib{X_T \geq U}  \ee^{\ii u (X_T - U)} , \label{eq:psiU.sbki}
\end{align}
and, if $\tau_U \leq T$, exchange the European claim at time $\tau_U$ for the knock-in claim at no cost.
\end{remark}

\begin{proposition}[Prices of single barrier knock-in claims on fractional powers of quadratic variation]
\label{prp:frac.sbki}
For any $0<r<1$ and $L < X_0$ we have 
\begin{align}
\Eb \Ib{\tau_L \leq T}( \<X\>_T - \<X\>_{\tau_L^*} )^r 
	&= \Eb g(X_T) , 
\end{align}
where 
\begin{align}
g(x)
	&:=	\frac{r}{\Gam(1-r)} \int_0^\infty \frac{1}{z^{r+1}} 
			\Big( \psi_L^{\ki}(x;0,0) - \psi_L^{\ki}(x;0,\ii z) \Big) \dd z . \label{eq:g.frac.sbki}
\end{align}
Here, $\Gam$ is the Euler Gamma function and $\psi_L^\ki(x;\om,s)$ is defined in \eqref{eq:psi.sbki}.
\end{proposition}

\begin{proof}
Following the proof of \cite[Proposition 7.1]{rrvd}, we have from \cite[equation (1.2.3)]{schurger2002laplace} that
\begin{align}
v^r
	&=	\frac{r}{\Gam(1-r)} \int_0^\infty \frac{1-\ee^{-z v} }{z^{r+1}} \dd z, &
v
	&\geq 0 , &
0
	&<r<1 . \label{eq:frac}
\end{align}
Therefore
\begin{align}
&\Eb \Ib{\tau_L \leq T} (\<X\>_T - \<X\>_{\tau_L^*})^r \\
	&=	\frac{r}{\Gam(1-r)} \int_0^\infty  \frac{1}{z^{r+1}} \Eb \Ib{\tau_L \leq T} \Big( 1-\ee^{-z (\<X\>_T - \<X\>_{\tau_L^*})} \Big) \dd z &
	&\text{(by \eqref{eq:frac} and Tonelli)}\\
	&=	\frac{r}{\Gam(1-r)} \int_0^\infty \frac{1}{z^{r+1}} \Eb \Big( \psi_L^{\ki}(X_T;0,0) - \psi_L^{\ki}(X_T;0,\ii z) \Big) \dd z &
	&\text{(by Theorem \ref{thm:sbki})}\\
	&=	\Eb g(X_T)  &
	&\text{(by \eqref{eq:g.frac.sbki} and Fubini)}
\end{align}
where the use of Fubini is justified as follows.  As $z\to \infty$ we have $\ii u(\om,\ii z) \to 1/2$ by \eqref{eq:psi.sbki}, hence 
\begin{align}
\Eb | \psi_L^{\ki}(X_T;0,0) - \psi_L^{\ki}(X_T;0,\ii z) | 
	&= \Oc(1) , &
	&\text{as $z \to \infty$.} \label{eq:sbki.0}
\end{align}
Turning to the $z\to 0$ behavior, consider the case $u = u_+$; the proof for $u=u_-$ is similar.
We have
\begin{align}
&\big| \psi_L^\ki(X_T;0,0) - \psi_L^\ki(X_T;0,\ii z) \big| \\
	&\leq \Ib{X_T < L}\ee^{X_T-L} \big|	1 - \ee^{-\ii u(0,\ii z) (X_T - L)} \big| + \Ib{X_T \leq L} \big|	1 - \ee^{\ii u(0,\ii z) (X_T - L)} \big| \\
	&=	\Ib{\tau_L \leq T} \Big(
			\Ib{X_T < L}\ee^{X_T-L} \big|	1 - \ee^{-\ii u(0,\ii z) (X_T - X_{\tau_L^*})} \big| + \Ib{X_T \leq L} \big|	1 - \ee^{\ii u(0,\ii z) (X_T - X_{\tau_L^*})} \big|
			\Big) \\
	&\leq \Ib{\tau_L \leq T} \big|	1 - \ee^{-\ii u(0,\ii z) (X_T - X_{\tau_L^*})} \big| 
			+ \Ib{\tau_L \leq T} \big|	1 - \ee^{\ii u(0,\ii z) (X_T - X_{\tau_L^*})} \big| . \label{eq:sbki.1}
\end{align}
using $u(0,0)=0$ as well as $\Ib{\tau_L > T}\Ib{X_T \leq L} = 0$ and $\Ib{X_T < L}\ee^{X_T-L} \leq 1$.
For $z$ small enough, we have $\ii u(0,\ii z) = 1/2 - \sqrt{1/4 - 2z} \in \Rb$. Thus
\begin{align}
\Big( \Eb \Ib{\tau_L \leq T} \big|	1 - \ee^{-\ii u(0,\ii z) (X_T - X_{\tau_L^*})} \big| \Big)^2
	&\leq \Eb \Ib{\tau_L \leq T} \big|	1 - \ee^{-\ii u(0,\ii z) (X_T - X_{\tau_L^*})} \big|^2 \\
	&=	\Eb \Ib{\tau_L \leq T} \Eb_{\tau_L^*}\big(	1 - 2 \ee^{-\ii u(0,\ii z) (X_T - X_{\tau_L^*})} + \ee^{-2 \ii u(0,\ii z) (X_T - X_{\tau_L^*})}\big) \\
	&=	\Eb \Ib{\tau_L \leq T} \Eb_{\tau_L^*}\big(	1 - 2 \ee^{-a(z) (\<X\>_T - \<X\>_{\tau_L^*})} + \ee^{-b(z) (\<X\>_T - \<X\>_{\tau_L^*})}\big) , \\
a(z)
	&= 	-\tfrac{1}{2} + \tfrac{1}{2} \sqrt{1-8z} + 2 z, \\
b(z)
	&=	-\tfrac{3}{2} + \tfrac{3}{2} \sqrt{1-8z} + 8 z ,
\end{align}
using $\ii u(0,\ii a(z)) = - \ii u(0,\ii z)$ and $\ii u(0,\ii b(z)) = - 2\ii u(0,\ii z)$.  
Now define $f(d):= \Eb \Ib{\tau_L \leq T} \ee^{-d (\<X\>_T - \<X\>_{\tau_L^*})}$.
The function $f$ is analytic by \eqref{eq:bound}.  Hence
\begin{align}
&\Big( \Eb \Ib{\tau_L \leq T} \big|	1 - \ee^{-\ii u(0,\ii z) (X_T - X_{\tau_L^*})} \big| \Big)^2 \\
	&\leq f(0) - 2f(a(z)) + f(b(z)) \\
	&=	\Big( f(0) - 2 f(a(0)) + f(b(0)) \Big) + \Big( - 2 f'(a(0)) a'(0)+ f'(b(0))b'(0) \Big) z + \Oc(z^2)
	=		\Oc(z^2) , \label{eq:sbki.2}
\end{align}
because $a(0) = b(0) = 0$ and $-2 a'(0) + b'(0) = 0$.  A similar computation shows
\begin{align}
\Big( \Eb \Ib{\tau_L \leq T} \big|	1 - \ee^{\ii u(0,\ii z) (X_T - X_{\tau_L^*})} \big| \Big)^2
	&=	\Oc(z^2) . \label{eq:sbki.3}
\end{align}
Combining \eqref{eq:sbki.1}, \eqref{eq:sbki.2} and \eqref{eq:sbki.3}, we have  
\begin{align}
\Eb \big| \psi_L^\ki(X_T;0,0) - \psi_L^\ki(X_T;0,\ii z) \big|
	&=	\Oc(z) &
	&\text{as $z \to 0$.} \label{eq:sbki.4}
\end{align}
From \eqref{eq:sbki.0} and \eqref{eq:sbki.4}, we have
\begin{align}
\frac{r}{\Gam(1-r)} \int_0^\infty \frac{1}{z^{r+1}} \Eb \Big| \psi_L^{\ki}(X_T;0,0) - \psi_L^{\ki}(X_T;0,\ii z) \Big| \dd z 
	&< \infty, 
\end{align}
which verifies the conditions of Fubini.
\end{proof}

\begin{proposition}[Ratio claims]
\label{prp:ratio.sbki}
For any $r,\eps > 0$ and $p \in \Cb$ we have
\begin{align}
\Eb \Bigg( \Ib{\tau_L\leq T} \frac{(X_T-X_{\tau_L^*}) \ee^{\ii p (X_T-X_{\tau_L^*})} }{ (\<X\>_T - \<X\>_{\tau_L^*}+ \eps)^{r} } \Bigg)
	&=	\Eb g(X_T) ,
\end{align}
where the function $g$ is given by
\begin{align}
g(x)
	&=	\frac{1}{r\Gam(r)}  \int_0^\infty (-\ii \d_p) \psi_L^\ki(x;p,\ii z^{1/r}) \ee^{- z^{1/r} \eps } \dd z. \label{eq:g.ratio.sbki}
\end{align}
Here, $\Gam$ is the Euler Gamma function and $\psi_L^\ki(x;\om,s)$ is as defined in \eqref{eq:psi.sbki}.
\end{proposition}

\begin{proof}
Following the proof of \cite[Proposition 7.2]{rrvd}, we have from \cite[equation (1.0.1)]{schurger2002laplace} that
\begin{align}
\frac{ 1 }{ v^r }
	&=	\frac{1}{r\Gam(r)} \int_0^\infty \ee^{ - z^{1/r} v } \dd z, &
r
	&>	0 . \label{eq:ratio}
\end{align}
from which
\begin{align}
&\Eb \Bigg( \Ib{\tau_L\leq T} \frac{(X_T-X_{\tau_L^*}) \ee^{\ii p (X_T-X_{\tau_L^*})} }{ (\<X\>_T - \<X\>_{\tau_L^*}+ \eps)^{r} } \Bigg) \\
	&=	\frac{1}{r\Gam(r)} \int_0^\infty \Eb \Ib{\tau_L\leq T} (X_T - X_{\tau_L^*} ) 
			\ee^{\ii p (X_T - X_{\tau_L^*} ) - z^{1/r} (\<X\>_T - \<X\>_{\tau_L^*}+ \eps)} \dd z  &
	&		\text{(by \eqref{eq:ratio} and Fubini)}\\
	&=	\frac{1}{r\Gam(r)} \int_0^\infty(-\ii \d_p) \Eb \Ib{\tau_L\leq T}
			\ee^{\ii p (X_T - X_{\tau_L^*} ) - z^{1/r} (\<X\>_T - \<X\>_{\tau_L^*}+ \eps)} \dd z &
	&		\text{(by Leibniz)}\\
	&=	\frac{1}{r\Gam(r)} \int_0^\infty (-\ii \d_p) \Eb \psi_L^\ki(X_T;p,\ii z^{1/r})\ee^{ - z^{1/r} \eps} \dd z &
	&		\text{(by Theorem \ref{thm:sbki})}\\
	&=	\Eb g(X_T) . &
	&		\text{(by Leibniz, \eqref{eq:g.ratio.sbki} and Fubini)}
\end{align}
The first use of Fubini is justified as, for all $p \in \Cb$ and $z \geq 0$, we have
\begin{align}
\Eb \big| \Ib{\tau_L\leq T} (X_T - X_{\tau_L^*} ) \ee^{\ii p (X_T - X_{\tau_L^*} ) - z^{1/r} (\<X\>_T - \<X\>_{\tau_L^*})} \big|
	&\leq \Eb \big| \Ib{\tau_L\leq T} (X_T - X_{\tau_L^*} ) \ee^{\ii p (X_T - X_{\tau_L^*} )} \big|
	 < \infty , 
\end{align}
from which 
\begin{align}
\int_0^\infty \Eb \Big| \Ib{\tau_L\leq T} (X_T - X_{\tau_L^*} ) \ee^{\ii p (X_T - X_{\tau_L^*} ) - z^{1/r} (\<X\>_T - \<X\>_{\tau_L^*}+ \eps)} \Big| \dd z 
	&< \infty .
\end{align}
The two uses of the Leibniz rule are by \eqref{eq:leibniz}.
The second use of Fubini is justified as follows.
By \eqref{eq:psi.sbki}, 
\begin{align}
&-\ii \d_p \psi_L^\ki(X_T;p,\ii z^{1/r}) \\
	&=	\Big( -\Ib{X_T < L} \ee^{(X_T - L) - \ii u(p,\ii z^{1/r}) (X_T - L)} + \Ib{X_T \leq L} \ee^{\ii u(p,\ii z^{1/r}) (X_T - L)} \Big) 
			\d_p u(p,\ii z^{1/r})(X_T - L) .
\end{align}
where, from \eqref{eq:u}, 
\begin{align}
\d_p u_\pm(p,\ii z^{1/r})
	&=	\frac{\pm (1-2 \ii p)}{\sqrt{1-4 p (p+\ii)-8 z^{1/r}}} . \label{eq:dpu}
\end{align}
As $\ii u(p,\ii z^{1/r}) \to 1/2$ and $\d_p u(p,\ii z^{1/r}) \to 0$ as $z \to \infty$, it follows that
\begin{align}
\Eb \big| -\ii \d_p \psi_L^\ki(X_T;p,\ii z^{1/r}) \big|
	&= \Oc(1) , &
	&		\text{as $z \to \infty$.}
\end{align}
Therefore
\begin{align}
\int_0^\infty\Eb \big| (-\ii \d_p) \psi_L^\ki(X_T;p,\ii z^{1/r}) \big| \ee^{ - z^{1/r} \eps}\dd z 
	< \infty , \label{eq:integral}
\end{align}
where any possible singularity in the integrand of \eqref{eq:integral} due to the denominator in \eqref{eq:dpu} will not cause the integral to explode as, for any $a \in \Rb_+$ we have
\begin{align}
\int_0^\infty\Big| \frac{\ee^{- \eps z^{1/r}}}{\sqrt{a-z^{1/r}}} \Big|\, \dd z 
	&=	\int_0^\infty r x^{r-1} \ee^{- \eps x} \Big| \frac{1}{\sqrt{a-x}} \Big|\, \dd x 
	< \infty ,
\end{align}
thus justifying the second use of Fubini.
\end{proof}

Figure \ref{fig:sbki} plots the European claims \eqref{eq:g.frac.sbki} and \eqref{eq:g.ratio.sbki} whose expectations price, respectively
\begin{align} 
\text{single barrier knock-in volatility swap}:&&
\Ib{\tau_L \leq T} \sqrt{\<X\>_T - \<X\>_{\tau_L^*}}, \label{eq:VolSwap} \\
\text{single barrier knock-in realized Sharpe ratio}:&&
\Ib{\tau_L \leq T} \frac{ X_T - X_{\tau_L^*} }{ \sqrt{\<X\>_T - \<X\>_{\tau_L^*} + \eps} } . \label{eq:SharpeRatio}
\end{align}

%
%



\section{Single barrier rebate claims}
\label{sec:sbr}
This section considers \emph{single barrier rebate} claims with payoffs of the form
\begin{align}
\text{Single barrier rebate}:&&
	 \Ib{\tau_H \leq T}\varphi(\<X\>_{\tau_H^*}) ,
\end{align}
paid at time $\tau_H^*$.
Define $v_\pm: \Cb \to \Cb$ by
\begin{align}
v_\pm(s) &:= \ii \Big( -\tfrac{1}{2} \pm \sqrt{ \tfrac{1}{4} - 2 \ii s} \Big).  \label{eq:vM}
\end{align}
%
As with $u_\pm(\om,s)$, when it causes no confusion, we will omit the subscript $\pm$ and the argument $s$ from $v_\pm(s)$.
The following trading strategy replicates a single barrier rebate power-exponential claim.

\begin{theorem}[Replication of single barrier rebate power-exponential claims]
\label{thm:sbr}
Fix $s \in \Cb\setminus\{-\ii/8\}$, $m \in \{0\} \cup \Nb$ and $H \in \Rb$.  
Define 
\begin{align}
\psi_H^\rb(X_T,\<X\>_T) \equiv \psi_H^\rb(X_T,\<X\>_T;s)
	&=	\ee^{\ii v (X_T - H) + \ii s \<X\>_T} , \label{eq:psi.rb}
\end{align}
where $v \equiv v(s)$ is defined in \eqref{eq:vM}.  
Then the following trading strategy replicates a single barrier rebate power-exponential claim that pays
\begin{align}
\Ib{\tau_H \leq T} \<X\>_{\tau_H^*}^m \ee^{\ii s\<X\>_{\tau_H^*} } , \label{eq:sbr}
\end{align}
at time $\tau_H^*$.
At time $0$ hold one European-style claim with payoff $(-\ii \d_s)^m \psi_H^\rb(X_T,\<X\>_T;s)$ and 
sell one single barrier knock-out claim with payoff $\Ib{\tau_H>T} (-\ii \d_s)^m \psi_H^\rb(X_T,\<X\>_T;s)$.
If and when $X$ hits the level $H$, the knock-out claim becomes worthless; sell the European-style claim for $\<X\>_{\tau_H^*}^m \ee^{\ii s \<X\>_{\tau_H^*}}$.
\end{theorem}

\begin{proof}
If $\tau_H > T$ the rebate claim expires worthless.  Likewise, if $\tau_H > T$, the long position in the European-style claim pays $(-\ii \d_s)^m\psi_H^\rb(X_T,\<X\>_T;s)$ while the short position in the single barrier knock-out claim pays $-(-\ii \d_s)^m\psi_H^\rb(X_T,\<X\>_T;s)$, for a net payout of zero.

If $\tau_H \leq T$, the knock-out claim becomes worthless at time $\tau_H^*$.  Thus it remains to show that the value of the European-style claim equals the payoff of the rebate claim at time $\tau_H^*$.  We have
\begin{align}
& \Ib{\tau_H \leq T} \Eb_{\tau_H^*} (-\ii \d_s)^m \psi_H^\rb(X_T,\<X\>_T;s) \\
	&=	\Ib{\tau_H \leq T} (-\ii \d_s)^m  \ee^{-\ii v H} \Eb_{\tau_H^*} \ee^{\ii v X_T + \ii s \<X\>_T} &
	&		\text{(by \eqref{eq:psi.rb} and Leibniz)} \\
	&=	\Ib{\tau_H \leq T} (-\ii \d_s)^m  \ee^{-\ii v H} \ee^{\ii v X_{\tau_H^*} + \ii s \<X\>_{\tau_H^*}} &
	&		\text{($M_t:=\ee^{\ii v X_t + \ii s \<X\>_t}=\Eb_t M_T$ is a martingale)} \\
	&=	\Ib{\tau_H \leq T} (-\ii \d_s)^m \ee^{\ii s \<X\>_{\tau_H^*}} &
	&		\text{($\Ib{\tau_H \leq T}(X_{\tau_H^*}-H) = 0$)} \\
	&=	\Ib{\tau_H \leq T} \<X\>_{\tau_H^*}^m \ee^{\ii s \<X\>_{\tau_H^*}} , \label{eq:proof.sbr}
\end{align}
where use of the Leibniz integral rule is justified by \eqref{eq:leibniz}.
\end{proof}

\begin{proposition}[Prices of single barrier rebate power-exponential claims]
\label{prp:sbrb-pow-exp}
Assume the distribution of $X_T$ has no point masses (a sufficient condition is that $\int_0^T \sig_t^2 \dd t > \eps > 0$).  Then for any $L < X_0$, $k \in \{0\} \cup \Nb$ and $s \in \Cb \setminus \{-\ii/8\}$, we have
\begin{align}
\Eb \Ib{\tau_L \leq T} \<X\>_{\tau_L^*}^k \ee^{\ii s \<X\>_{\tau_L^*}}
	&=	\lim_{n \to \infty} \Eb \Big( g_n(X_T) + h_n(X_T) \Big) , \label{eq:sbr-main}
\end{align}
where the functions $g_n$ and $h_n$ are given by
\begin{align}
g_n(X_T)
	&= \int_\Rb  (-\ii \d_s)^k \Hh_n(v(s)-\om) \ee^{-\ii \om L + \ii (\om-u(\om,s)X_0 + \ii u(\om,s) X_T}\dd \om_r , \\
h_n(X_T)
	&=	\int_\Rb(-\ii \d_s)^k \Hh_n(- \ii - v(s) - \om) \ee^{-\ii \om L + \ii (\om-u(\om,s)X_0 + \ii u(\om,s) X_T} \dd \om_r, 
\end{align}
with $\Hh_n$ as defined in \eqref{eq:Hhat}.
Here, the contour of integration in $g_n$ is chosen so that $2n + \Im v(s) > \om_i >  \Im v(s)$ and $2 \ii s-\omega ^2-\ii \om +\tfrac{1}{4} \neq 0$ and
the contour of integration in $h_n$ is chosen so that $2n - 1 - \Im v(s) > \om_i > - 1 - \Im v(s)$ and $2 \ii s-\omega ^2-\ii \om +\tfrac{1}{4} \neq 0$.
\end{proposition}

\begin{proof}
As in the proof of Proposition \ref{prp:sbko-pow-exp}, let $H(x):=\tfrac{1}{2}(1 + \text{sgn} \, x)$ denote the Heaviside function and let $H_n(x) := \tfrac{1}{2} ( 1 + \tanh n x )$.  Then
\begin{align}
&\Eb \Ib{\tau_L \leq T} \<X\>_{\tau_L^*}^k \ee^{\ii s \<X\>_{\tau_L^*}} \\
	&=	(-\ii \d_s)^k \Eb \Ib{\tau_L \leq T} \ee^{\ii s \<X\>_{\tau_L^*}} \\
	&=	(-\ii \d_s)^k \Eb\Big( \psi_L^\rb(X_T,\<X\>_T;s) - \Ib{\tau_L>T} \psi_L^\rb(X_T,\<X\>_T;s) \Big) \\
	&=	 (-\ii \d_s)^k \Eb \Big(  \psi_L^\rb(X_T,\<X\>_T;s) 
			- \Ib{X_T > L} \psi_L^\rb(X_T,\<X\>_T;s) + \Ib{X_T < L} \ee^{X_T - L} \psi_L^\rb(2L-X_T,\<X\>_T;s) \Big) \\
	&=	(-\ii \d_s)^k \Eb \Big( \Ib{X_T \leq L} \psi_L^\rb(X_T,\<X\>_T;s) + \Ib{X_T < L} \ee^{X_T - L} \psi_L^\rb(2L-X_T,\<X\>_T;s) \Big) \\
	&=	\Eb \Big(  \Ib{X_T \leq L} (-\ii \d_s)^k  \psi_L^\rb(X_T,\<X\>_T;s) + \Ib{X_T < L} \ee^{X_T - L} (-\ii \d_s)^k \psi_L^\rb(2L-X_T,\<X\>_T;s) \Big) \\
	&=	\Eb \lim_{n \to \infty } H_n(L-X_T) (-\ii \d_s)^k \Big( \psi_L^\rb(X_T,\<X\>_T;s) + \ee^{X_T - L} \psi_L^\rb(2L-X_T,\<X\>_T;s) \Big) \\
	&=	\lim_{n \to \infty } (-\ii \d_s)^k \Eb H_n(L-X_T) \Big( \psi_L^\rb(X_T,\<X\>_T;s) + \ee^{X_T - L} \psi_L^\rb(2L-X_T,\<X\>_T;s) \Big) , \label{eq:sbr-1}
\end{align}
where the second equality follows from Theorem \ref{thm:sbr},
the third equality follows from Theorem \ref{thm:sbko}, 
the fourth equality is algebra,
the sixth equality follows from the fact that the distribution of $X_T$ has no point masses (by assumption)
and the various exchanges of limits, derivatives and expectations are allowed by Lebesgue's dominated convergence and the Leibniz integral rule.
Using the expression \eqref{eq:psi.rb} for $\psi_L^\rb$ and the fact that $\Fv[H_n]=\Hh_n$ we have
\begin{align}
&(-\ii \d_s)^k \Eb H_n(L-X_T) \psi_L^\rb(X_T,\<X\>_T;s) \\
	&=	(-\ii \d_s)^k \Eb \int_\Rb \Hh_n(v(s) - \om) \ee^{-\ii \om L} \ee^{\ii \om X_T + \ii s\<X\>_T}\dd \om_r  &
	&		(2n + \Im v(s) > \om_i >  \Im v(s))  \\
	&=	(-\ii \d_s)^k \int_\Rb \Hh_n(v(s) - \om) \ee^{-\ii \om L}  \Eb \ee^{\ii \om X_T + \ii s\<X\>_T}\dd \om_r   &
	&		\text{(Fubini)}\\
	&=	(-\ii \d_s)^k \int_\Rb\Hh_n(v(s) - \om) \ee^{-\ii \om L + \ii (\om-u(\om,s)X_0} \Eb \ee^{\ii u(\om,s) X_T} \dd \om_r &
	&		\text{(by \eqref{eq:E1=E2})}\\
	&=	(-\ii \d_s)^k \Eb \int_\Rb\Hh_n(v(s) - \om) \ee^{-\ii \om L + \ii (\om-u(\om,s)X_0 + \ii u(\om,s) X_T}\dd \om_r  &
	&		\text{(Fubini)}\\
	&=	\Eb g_n(X_T) , \label{eq:sbr-2}
\end{align}
where the two applications of Fubini's theorem and the use of the Leibniz integral rule are justified by the 
$|\om_r| \to \infty$ behaviors: 
$|\d_s^k \Hh_n(v(s)-\om)| = \Oc( \ee^{ - |\om_r|/n })$
and
\begin{align} 
\Eb |\ee^{-\ii \om L} \ee^{\ii \om X_T + \ii s\<X\>_T}| 
	&= \Oc(1)  , &
\Eb |\d_s^k \ee^{-\ii \om L + \ii (\om-u(\om,s)X_0 + \ii u(\om,s) X_T}| 
	&= \Oc(1) ,
\end{align}
and by the choice of the integration contour to avoid any singularities in the integrand.   Similarly,
\begin{align}
(-\ii \d_s)^k \Eb H_n(X_T-L) \ee^{X_T-L}\psi_L^\rb(2L-X_T,\<X\>_T;s) 
	&=	\Eb h_n(X_T). \label{eq:sbr-3}
\end{align}
The result \eqref{eq:sbr-main} follows from \eqref{eq:sbr-1}, \eqref{eq:sbr-2} and \eqref{eq:sbr-3}.
\end{proof}

\begin{remark}
\label{rmk:sbrb}
To price a single-barrier rebate power-exponential claim $\Ib{\tau_U \leq T} \<X\>_{\tau_U^*}^k \ee^{\ii s \<X\>_{\tau_U^*}}$ with up-barrier $U > X_0$, make the following changes to Proposition \ref{prp:sbrb-pow-exp}: replace
\begin{align}
L 
	&\to U , &
\Hh(v(s)-\om) 
	&\to \Hh(\om-v(s)) , &
\Hh(-\ii - v(s) - \om) 
	&\to \Hh(\om + \ii + v(s))
\end{align}
and reflect the contours of integration over the real axis: $\om_i \to -\om_i$.
\end{remark}

In Figure \ref{fig:sbrb}, for various values of $X_0$ and $H$, we plot the payoff \eqref{eq:sbr-main} whose expectation, in the limit as $n \to \infty$, prices the single barrier rebate variance swap which pays $\Ib{\tau_H \leq T} \<X\>_{\tau_H^*}$.

%
%

\section{Summary and future research}
\label{sec:conclusion}
Assuming only that the price of a risky asset $S=\ee^X$ is strictly positive and continuous and driven by an independent volatility process $\sig$, we have shown how to price and hedge a variety of barrier-style claims written on the $\log$ returns $X$ and the quadratic variation of $\log$ returns $\<X\>$.  In particular, we have studied single and double barrier knock-in, knock-out, and rebate claims.  The pricing formulas we obtain are semi-robust in that they do not specify the 
dynamics of $\sigma$, which may be non-Markovian and may include jumps. 	
\par
Future research will focus three areas
(i) weakening the independence assumption on $\log$ returns and volatility,
(ii) pricing and hedging when calls and puts are available only at discrete strikes or only within a finite interval,
(iii) considering richer payoff structures, which may depend on the running maximum or minimum of the asset in addition to $\log$ returns and quadratic variation of $\log$ returns.

%
%

\bibliographystyle{chicago}
\bibliography{Bibtex-Master-3.05}	

%
%

\clearpage

\begin{figure}
\centering
\begin{tabular}{ | c | c |}
\hline
\includegraphics[width=0.45\textwidth]{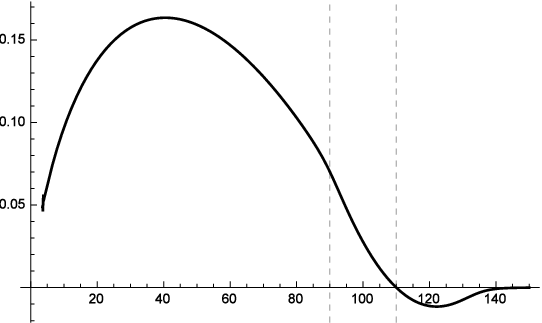} &
\includegraphics[width=0.45\textwidth]{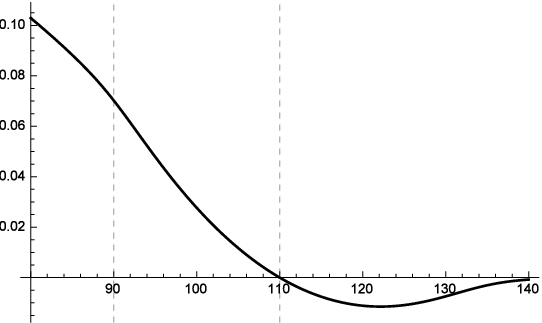} \\ \hline
\end{tabular}
\caption{
We plot the payoff $g_n(\log \cdot) - h_n(\log \cdot)$, whose expectation, in the limit as $n \to \infty$, prices a single barrier knock-out power-exponential claim; see equation \eqref{eq:sbko-main}.  In both plots, the following parameters are fixed: $L = \log 90$, $X_0 = \log 110$, $p=0$, $s=0$, $j=0$, $k=1$, $n = 25$.  
The vertical dashed lines are placed at $\ee^L = 90$ and $S_0 = \ee^{X_0} = 110$.  
Note that with $(p,s,j,k)$ as chosen, the European payoff function plotted above has expectation which, in the large $n$ limit, prices a single barrier knock-out variance swap, which pays $\Ib{\tau_L > T} \<X\>_T$.
}
\label{fig:sbko}
\end{figure}

\begin{figure}
\centering
\begin{tabular}{ | c | c |}
\hline
\includegraphics[width=0.45\textwidth]{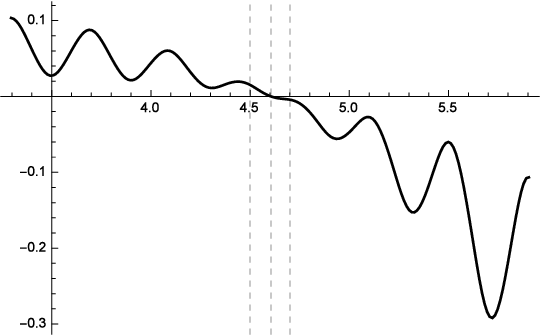} &
\includegraphics[width=0.45\textwidth]{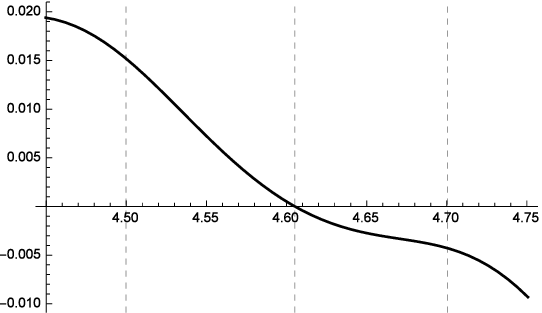} \\ \hline
\end{tabular}
\caption{
We plot the payoff appearing on the right-hand side of \eqref{eq:dbko-pow-exp}, whose expectation, in the limit as $q,m \to \infty$, prices a double barrier knock-out power-exponential claim.  In both plots, the following parameters are fixed: $L = \log 90$, $U = \log 110$, $X_0 = \log 100$, $p=0$, $s=0$, $j=0$, $k=1$, $m = 15$ and $q=5$.
The vertical dashed lines are placed at $L = 90$, $X_0 = 100$ and $U = \log 110$.  
Note that with $(p,s,j,k)$ as chosen, the European payoff function plotted above has expectation which, in the large $q,m$ limit, prices a double barrier knock-out variance swap, which pays $\Ib{\tau_{L,U} > T} \<X\>_T$.
}
\label{fig:dbko}
\end{figure}

\begin{figure}
\centering
\begin{tabular}{ | c | c |}
\hline
\includegraphics[width=0.45\textwidth]{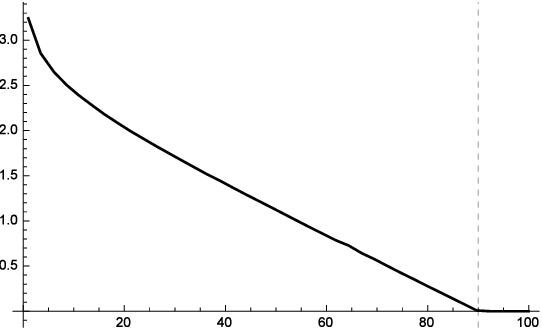} &
\includegraphics[width=0.45\textwidth]{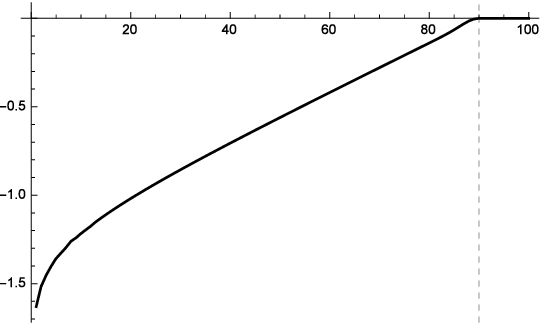} \\ \hline
\end{tabular}
\caption{
\emph{Left}: we plot the payoff $g(\log \cdot)$, given in \eqref{eq:g.frac.sbki}, whose expectation prices a single-barrier knock-in claim on volatility with payoff \eqref{eq:VolSwap}.
\emph{Right}: we plot the payoff $g(\log \cdot)$, given in \eqref{eq:g.ratio.sbki}, whose expectation prices a single-barrier knock-in claim on realized Sharpe ratio with payoff \eqref{eq:SharpeRatio}.
In both plots, the following parameters are fixed: $\ee^L = 90$ and $r=1/2$.  For the ratio claim, we have additionally fixed $\eps = 0.001$.
The vertical line in both plots are placed at the knock-in barrier $\ee^L = 90$.
}
\label{fig:sbki}
\end{figure}

\begin{figure}
\centering
\begin{tabular}{ | c | c |}
\hline
\includegraphics[width=0.45\textwidth]{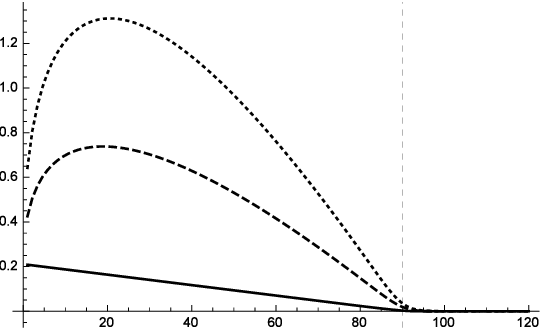} &
\includegraphics[width=0.45\textwidth]{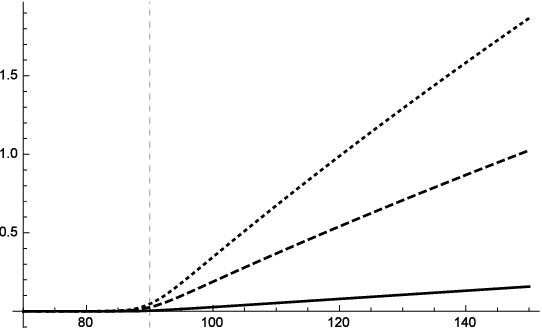} \\ \hline
\end{tabular}
\caption{
We plot the payoff $g_n(\log \cdot) + h_n(\log \cdot)$, given by \eqref{eq:sbr-main} (see also Remark \ref{rmk:sbrb}) whose expectation, in the limit as $n \to \infty$, prices a single barrier rebate power-exponential claim.
\textit{Left}: the solid, dashed, and dotted lines correspond to $\ee^{X_0} = \{100,100^{1.25},100^{1.50}\}$, respectively, and the other parameters are fixed: $\ee^L = 90$, $s=0$, $k=1$, $n = 25$.  The vertical dashed line is placed at the barrier $\ee^L = 90$.
\textit{Right}: the solid, dashed, and dotted lines correspond to $\ee^{X_0} = \{80,80^{2/3},80^{1/3}\}$, respectively, and the other parameters are fixed: $\ee^U = 90$, $s=0$, $k=1$, $n = 25$.  The vertical dashed line is placed at the barrier $\ee^U = 90$.
Note that with $(k,s)$ as chosen, the European payoffs plotted above have expectations which price, in the limit as $n \to \infty$, single barrier rebate variance swaps, all of which have a payoff $\Ib{\tau_H \leq T} \<X\>_{\tau_H^*}$.  
}
\label{fig:sbrb}
\end{figure}

\end{document}